\documentclass[11pt]{article}
\usepackage{amsmath,amsthm,amscd,amssymb}
\usepackage[latin1]{inputenc}
\usepackage[a4paper,left=2.5cm,right=2.5cm,top=2.5cm,bottom=2.5cm]{geometry}
\usepackage{amsfonts}
\usepackage{cancel}
\usepackage[bookmarks=false,colorlinks=true,linkcolor=black,citecolor=black,filecolor=black,urlcolor=black]{hyperref}
\usepackage[usenames,dvipsnames]{color}
\usepackage[dvips]{graphicx}
\usepackage{natbib}

\usepackage{multirow}
\usepackage{bigstrut}
\usepackage{tabulary}
\usepackage{lscape}
\usepackage{array}
\usepackage{rotating}
\usepackage[normalem]{ulem}
\usepackage[displaymath,mathlines]{lineno}
\usepackage[labelfont=bf]{caption}

\newtheorem{theorem}{Theorem}[section]
\newtheorem{lemma}[theorem]{Lemma}

\newcommand{\matriz}[1]{\begin{array} #1 \end{array}}

\newcommand{\MC}[1]{\mathcal{#1}}

\newcommand{\A}{\mathcal{AP}}
\newcommand{\AL}{\mathcal{LAP}}
\newcommand{\D}{\mathcal{D}}
\newcommand{\B}{{\rm B}}
\newcommand{\Beta}{{\rm Beta}}

\newcommand{\MAF}{{\rm MAF}}

\usepackage{authblk}

\title{Approaching allelic probabilities and Genome-Wide Association Studies from beta distributions.}

\author[1]{J. Santiago Garc\'{i}a-Cremades\thanks{Correspondence to: Jos\'e Santiago Garc\'ia Cremades, Departamento de Matem\'aticas, Universidad de Murcia, 30100, Espinardo, Spain. E-mail: js.garciacremades@gmail.com}}
\author[1]{\'{A}ngel del R\'{i}o}
\author[2]{Jos\'{e} A. Garc\'{i}a}
\author[3,4]{Javier Gay\'{a}n}
\author[3,5]{Antonio Gonz\'{a}lez-P\'{e}rez}
\author[3,6]{Agust\'{i}n Ruiz}
\author[6]{Oscar Sotolongo-Grau}
\author[2]{Manuel Ruiz-Mar\'{i}n}

\affil[1]{Departamento de Matem\'aticas, Universidad de Murcia.}
\affil[2]{Departamento de M\'etodos Cuantitativos e Inform\'aticos, Universidad Polit\'ecnica de Cartagena.}
\affil[3]{Department of Estructural Genomics, Neocodex, Sevilla, Spain.}
\affil[4]{Bioinfosol, Sevilla, Spain.}
\affil[5]{Centro Andaluz de Estudios Bioinformáticos (CAEBi), Sevilla, Spain.}
\affil[6]{Memory Clinic of Fundaci\'o ACE. Institut Català de Neurociències Aplicades. Barcelona. Spain.}

\providecommand{\keywords}[1]{\textbf{\textit{Keywords:}} #1}

\date{}
\makeatletter
\renewcommand\@biblabel[1]{}
\makeatother

\begin{document}

\maketitle
\newpage

\begin{abstract}
In this paper we have proposed a model for the distribution of allelic probabilities for generating populations as reliably as possible. Our objective was to develop such a model which would allow simulating allelic probabilities with different observed truncation and degree of noise. In addition, we have also introduced here a complete new approach to analyze a genome-wide association study (GWAS) dataset, starting from a new test of association with a statistical distribution and two effect sizes of each genotype. The new methodological approach was applied to a real data set together with a Monte Carlo experiment which showed the power performance of our new method. Finally, we compared the new method based on beta distribution with the conventional method (based on Chi-Squared distribution) using the agreement Kappa index  and a principal component analysis (PCA). Both the analyses show found differences existed between both the approaches while selecting the single nucleotide polymorphisms (SNPs) in association.\\
\end{abstract}

\keywords{GWAS; case-control study; allelic probability; beta distribution}

\section{Introduction}

A fundamental point in a genome-wide association study (GWAS) is to define a model for allelic probabilities. A good model must strongly depend on the observed truncation and the degree of noise that distort the observed (empirical) distribution from the expected one.
By controlling the truncation and the degree of noise, allelic probabilities can be simulated in a more reliable scenario.

The Human Genome Project (\cite{Lander2001}) and the successive improvements on physical and genetic maps of the human genome have boosted the post-genome era (\cite{Altshuler2010, Frazer2007, Sachidanandam2001}). Concomitant technological achievements in the genetic field have been successfully applied to uncover thousands of genetic variants linked to multiple phenotypes. The deep characterization of loci involved in both Mendelian and complex disorders will further help in improving diagnostic resolution and, ultimately, provide clues on the design of next generation therapeutics based on etiology, rather than in symptoms or clinical findings.

The genome-wide association studies (GWAS) appear to be unstoppable. The development of high density genome-wide panels of single nucleotide polymorphism (SNPs) and its application to bio-banked samples, accumulated during the last century, are the key elements explaining GWAS emergence. Till, date ongoing GWAS projects have published 1350 GWAS documents and more than 1800 GWAS significant loci (\url{www.genome.gov/gwastudies}; Freeze Dec/2012) (\cite{Hindorff2009}). The new loci have been detected using relatively well standardized methods based on linear additive models with or without covariants, a case-control design, by applying extensive quality control to raw genotyping data, by increasing density of markers based on inference of many non-genotyped markers using high performance computation (HPC), imputation techniques and by increasingly improved reference panels of single nucleotides polymorphisms (SNPs) (\cite{Bakker2008}).

In spite of these successes, most GWAS findings, typically of small effect sizes, leave a large fraction of disease susceptibility still unexplained; a phenomenon commonly known as ``the case of the missing heritability'' (\cite{Meesters2012}). Several potential explanations for this phenomenon were proposed (\cite{Manolio2009}). An excessive simplification of statistical methods applied to GWAS datasets might account for this problem. In this regard, allelic and additive models pervasively applied to GWAS data could be the genuine spherical cow (\cite{SheltonR2007}) on genetic research. Therefore, it is necessary to perform a continuous re-analysis and re-cycling of GWAS data by applying novel statistical methods to uncover those loci that match poorly with linear models (see \cite{Ruiz2010}).

Here, we have proposed a model for the distribution of allelic probabilities which allows to simulate allelic probabilities with different observed truncation and degree of noise. We have also introduced a complete new approach to GWAS analysis, starting from a new test for association and two effect sizes of each genotype. The new methodological approach was applied to a real data set together with a Monte Carlo experiment which showed the power performance of our new method and compared the new method with the conventional method.


\section{Materials and Methods}

\subsection{Modeling Allelic Probability Distribution}\label{SectionAllelicProbabilityCaseControl}


Allele frequency refers to the proportion of a certain allele on a genetic locus of population. These proportions often exhibit extra variation that cannot be explained by a simple binomial distribution. The proportion (or binomial parameter $p$) does not remain constant in the course of collecting data. Considering the situation, it would be useful to assume that the binomial parameter $p$ varies between observations. The data could be described assuming one of many continuous distributions for $p$, $0<p<1$. However, the most sensible distribution for $p$ is the beta distribution, because it is the natural conjugate prior distribution in the Bayesian sense.

It is well known that most SNPs present a very low minor allele frequency (MAF), close to either 0 or 1 depending on the codification.
However, these SNPs with very low MAF are systematically excluded from GWAS. After elimination, it is often assumed that the allele frequencies follow a uniform distribution. This is consistent with the beta approach, because the uniform distribution in the interval [0,1] is the beta distribution with parameters $\alpha=\beta=1$. However, the observed distribution of allelic frequencies is not quite uniform. Therefore, a finer analysis of the observed distribution is required.

The beta distribution with parameters $\alpha$ and $\beta$ is denoted by $\Beta(\alpha,\beta)$ and has the following probability density function (PDF)
    $$f(a)=\left\{\matriz{{ll} \frac{\B_{\alpha,\beta}(a)}{\int_0^1 \B_{\alpha,\beta}(r)dr}, & \text{if } 0<a<1;\\0, & \text{otherwise}.} \right.$$
where $\B_{\alpha,\beta}(a)=a^{\alpha-1}(1-a)^{\beta-1}$  and $0<\alpha, \beta\leq 1$.
The mean $\mu$ and variance $\sigma^2$ of $\Beta(\alpha,\beta)$ are given by
\begin{equation}\label{MediaVarianzaalphabeta-ok}
\mu=\frac{\alpha}{\alpha+\beta}, \quad
\sigma^2=\frac{\alpha\beta}{(\alpha+\beta)^2(\alpha+\beta+1)}.
\end{equation}
Notice that the parameters $\alpha$ and $\beta$ are determined by the mean and variance (\ref{MediaVarianzaalphabeta-ok}) by the following formulae:
    \begin{equation}\label{alphabetaMediaVarianza-ok}
    \alpha = \frac{\mu(\mu-\mu^2-\sigma^2)}{\sigma^2}, \quad \beta = \frac{(1-\mu)(\mu-\mu^2-\sigma^2)}{\sigma^2}.
   \end{equation}

In order to define variables, we assumed the allelic probability as the probability that one of the possible alleles occurs. For this, we denoted the two most frequent alleles as $A$ and $B$, respectively. This symbolic association was performed at random with equal probability. The allelic probability is defined as the probability of the occurrence of the allele denoted as $A$. This introduces a random variable, referred as \emph{allelic probability} ($\A$).

As explained above, since the values of $\A$ are proportions, the random variable $\A$ can be modeled with a beta distribution $\Beta(\alpha,\beta)$. As the chosen allele is determined at random the mean of $\A$ should be $0.5$. In terms of the beta distribution this implies that $\alpha=\beta$.

However, the beta distribution is not enough to properly model the allelic probability distribution in a real dataset. There are several considerations that we must take into account to explain this situation. For example, the commercial genome-wide SNP chips designs or quality controls (QC) applied to the genotyping studies, regularly exclude the SNPs with very small MAF, those with Hardy Weinberg disequilibrium, and those with a poor quality, etc. Furthermore, the incorporation of imputation methods may also introduce bias in genome-wide allelic distribution by imputing preferentially those SNPs with a higher MAF and those located in regions in strong linkage disequilibrium. Thus, one cannot assume that $\A$ takes all the values in the interval $\left[0, 1\right]$. Either design-based or QC-based pruning of SNPs induce a truncation of beta distribution.

Therefore, we may consider the random variable $\A_t$ of the truncated allelic probabilities in a GWAS with the truncation $t$, following a beta distribution truncated in some interval $[t,1-t]$, denoted as $\Beta(\alpha,\alpha,t,1-t)$. Its probability density function is given by:

\begin{equation}\label{betru-ok}
g(a)=\left\{\matriz{{ll} \frac{\B_{\alpha,\alpha}(a)}{\int_t^{1-t} \B_{\alpha,\alpha}(r)dr}, & \text{if } t<a<1-t;\\&\\0, & \text{otherwise}.} \right.
\end{equation}

Moreover, universal or general truncation exists in commercial chips that have been designed based on data derived from worldwide populations as HapMap or 1000 genome project datasets. In contrast, the DNA samples for a specific study are obtained from local populations. It is well established that SNP frequencies vary geographically due to genetic drift and exceptionally, due to natural selection in specific endemic regions. Therefore, a low MAF in studied local population might not indicate necessarily an equivalent low MAF in worldwide populations. As a consequence, selected SNPs in a commercial chip might not be observed in the local population, yielding a very low MAF that displays allelic probability values close to 0 or 1. To further analyze this phenomenon, we introduced the random variable a priori truncated local allelic probability, $\AL_t$, and the difference $\D_t=\AL_t-\A_t$, which measures the variation between the universal allelic probability and the local allelic probability of a given SNP.

\subsubsection{The local allelic probabilities}

We assumed that $\D_t$ and $\A_t$ are independent random variables. To guarantee local allelic probabilities to be bound into the interval $[0,1]$ (see Appendix \ref{Appendix-Variation}), we assumed that $\D_t$ is truncated in the interval $[-t,t]$. Hence we modeled $\D_t$ as a truncated normal distribution $NT(0,\delta,-t,t)$. Its probability density function is given by:
\begin{equation}
h(x)=\left\{\matriz{{ll} \frac{n_{0,\delta}(x)}{\int_{-t}^{t} n_{0,\delta}(r)dr}, & \text{if } -t<x<t;\\&\\0, & \text{otherwise}.} \right.
\end{equation}
where $n_{0,\delta}(x)=\frac{1}{\sqrt{2\pi}}e^{\frac{x^2}{-2\delta^2}}$ is the probability density function of a Normal distribution with mean zero and standard deviation $\delta$. Notice that the mean of $\D_t$ is equal to zero and its variance is denoted by $\sigma_D^2$.

As $\A_t$ and $\D_t$ were assumed to be independent. The probability density function of
\begin{equation}\label{lap-model}
\AL_t=\A_t+\D_t
\end{equation}
is the convolution of the probability density functions of $\A_t$ and $\D_t$.


As $\A_t$ is distributed as $\Beta(\alpha,\alpha,t,1-t)$, it is straightforward that its mean should be 0.5.
The variance of $\A_{t}$ was calculated using the incomplete beta function and the regularized incomplete beta function as shown in Appendix \ref{Appendix-Lema} (\ref{BetaIncompleta}).
On the other hand, as $\D_t$ follows a $NT(0,\delta,-t,t)$ distribution, its mean value is zero and its variance is given by (\ref{sigma_d_eq}).
    \begin{equation}\label{sigma_d_eq}
    \sigma_D^2=\delta^2\left( 1-\frac{2\frac{t}{\delta} \phi(\frac{t}{\delta})}{2\Phi(\frac{t}{\delta})-1} \right),
    \end{equation}
where $\phi$ and $\Phi$ are the PDF and CDF of the standard normal distribution $N(0,1)$ respectively. Moreover, we have $0\le \sigma_D^2 \le \frac{t^2}{3}$ (see Lemma~\ref{1/3}).

As $\A_t$ and $\D_t$ are independent, the variance of $\AL=\A_t+\D_t$ is $\sigma_l^2=\sigma_u^2+\sigma_D^2$ and hence
    $$\sigma_u^2 \le \sigma_l^2\le \sigma_u^2 +\frac{t^2}{3}.$$

This last expression includes the variance, taking into the account the noisy data. Indeed, if a higher degree of noise exists, the $\D_t$ variance can be taken as its maximum value, $\sigma_D^2 = \frac{t^2}{3}$. So, in this case, the expression $\sigma_l^2 \approx \sigma_u^2 +\frac{t^2}{3}$ can be used as a good approximation for $\AL$ variance.

\subsubsection{A real data example}\label{SectionRealData}

 In order to illustrate our approach we took the data from a practical case. The GWAS dataset is an imputed GWAS with 1,237,567 SNPs and 1225 individuals genotyped by the Translational Genomics Research Institute (TGEN) (\cite{Reiman2007}), previously processed as part of a genome-wide meta-analysis looking for Alzheimer's disease genetic risk factors (\cite{Antunez2011}).

 Figure~\ref{ejem1} shows the comparison between the empirical allele frequencies with the PDF of the uniform (A), beta (B) and truncated beta (C) distributions, where the empirical allele frequencies are represented by dots and the PDFs are represented by continuous lines. The parameters for the beta and truncated beta distribution were taken for the distribution to have the same mean and variance as the empirical data. Figure~\ref{ejem1} shows that the beta distribution is not good enough to model the empirical allelic frequencies. The observed divergence mainly occurs in the two tails of the distribution where the theoretical beta distribution increases, while the data dramatically decreases. The observed truncation is attributed to the removal of low MAF alleles during QC and the relative inability of imputation methods to make good inferences for low MAF SNPs.

\begin{figure}[h!]
\begin{center}
\begin{tabular}{ccc}
\includegraphics[height=3.5cm]{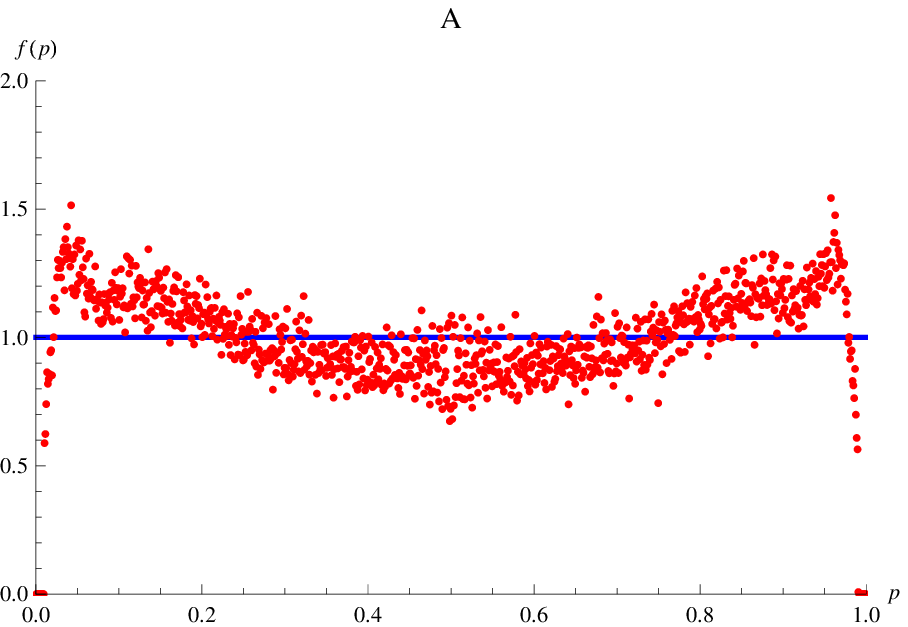} &
\includegraphics[height=3.5cm]{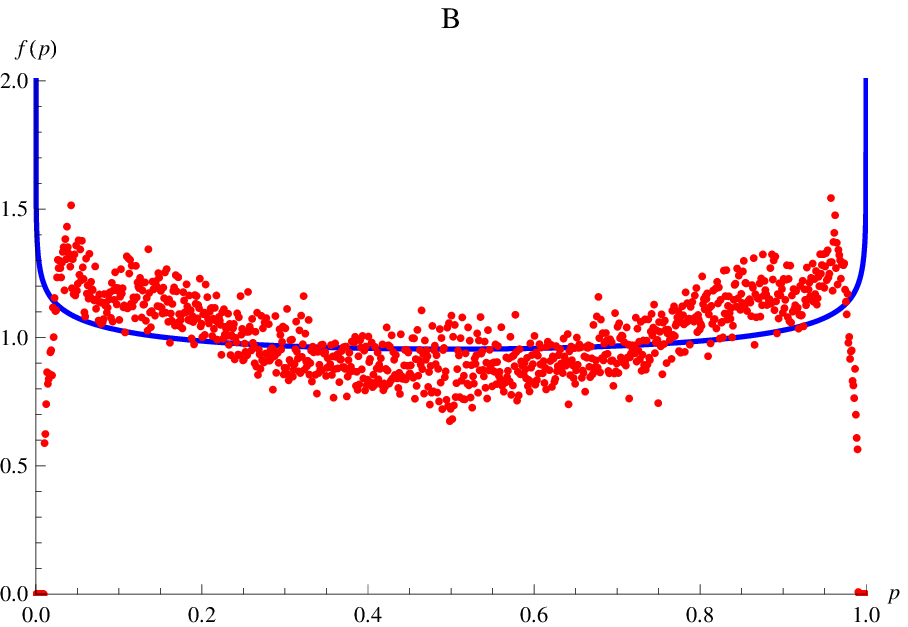} &
\includegraphics[height=3.5cm]{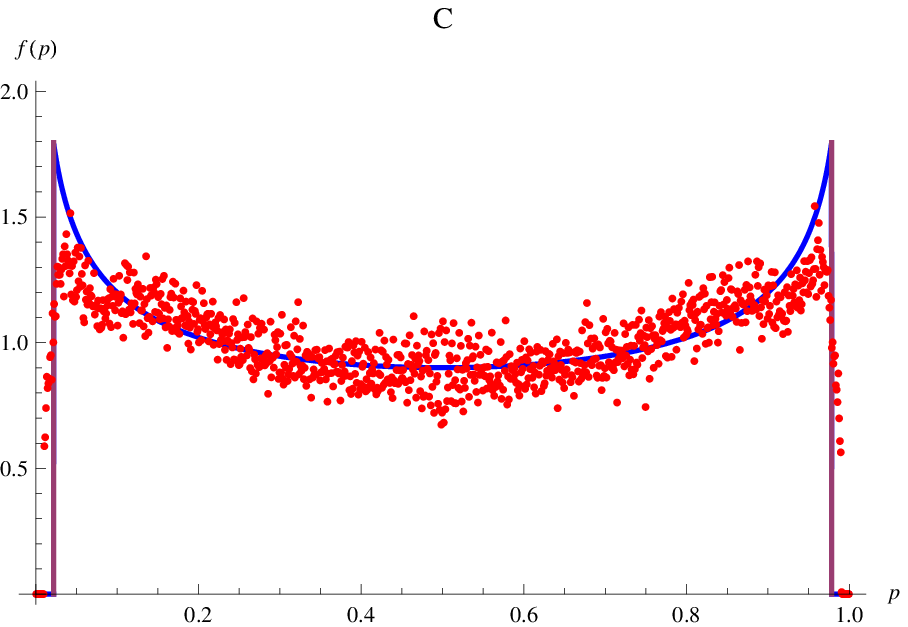}
\end{tabular}
\caption{\label{ejem1}Empirical relative frequencies (dots) and uniform (A), beta (B), and truncated beta (C) PDFs (continuous lines). 
(A) PDF of uniform distribution in $[0,1]$. (B) PDF of $\Beta\left(\alpha,\beta\right)$. (C) PDF of $\Beta\left(\alpha,\beta,t,1-t\right)$ . The parameters $\alpha$, $\beta$ of the beta distributions are taken to fit the mean and variance of empirical relative frequencies ($\alpha=\beta=0.928$) and the truncation is computed as the percentile 1 ($t=0.0218$). The empirical frequencies have been calculated grouping the data in 1000 intervals.}
\end{center}
\end{figure}

 In order to fit the data to equation (\ref{lap-model}), we took $t$ as the percentile 1 of the allelic probabilities: $t=0.0218$. Since  $\sigma_D^2$ belongs to the interval $[0,\frac{t^2}{3}]$ so that we can estimate $\sigma_D^2$ as if it is a uniform random variable in such an interval. This means that the expected value of $\sigma_D^2$ is calculated as $\sigma_D^2=\frac{t^2}{6}=7.9\cdot 10^{-5}$.
 Then, the parameter $\delta$ can be estimated as $\delta=9.59\cdot 10^{-3}$ from \eqref{sigma_d_eq}.

Since $t$ is known, $\sigma_l$ can be estimated from the sample. As $\sigma_u$ only depends on the parameters $\alpha$ and $t$, the value of $\alpha$ can be computed by solving (see Appendix \ref{Appendix-Lema}),

\begin{equation}
 \sigma_{u}^{2}=\sigma_{\alpha,t}^{2}= \frac{1}{(4\alpha + 2)\B(\alpha,\alpha)} \frac{t^{\alpha} (1-t)^{\alpha} (4t-2)}{1 - 2 I_t (\alpha,\alpha)} + \frac{ \alpha +1}{4\alpha + 2} -\frac{1}{4}=\sigma_{l}^{2}-\frac{t^{2}}{6},
\end{equation}
where $\sigma_{\alpha,t}^{2}$ is the variance of the truncated beta distribution. Therefore, $\alpha=0.7199$ is obtained from this equation. The three parameters, $t$, $\delta$, and $\alpha$ determine the $\AL_t$ distribution fitting to the dataset. Figure~\ref{ejem2} shows the empirical relative frequencies (dots) and the model of allelic probabilities $\AL_t$ (continuous line) with the parameters computed above.

\begin{figure}[b!]
\begin{center}
\begin{tabular}{c}
   \includegraphics[height=7cm]{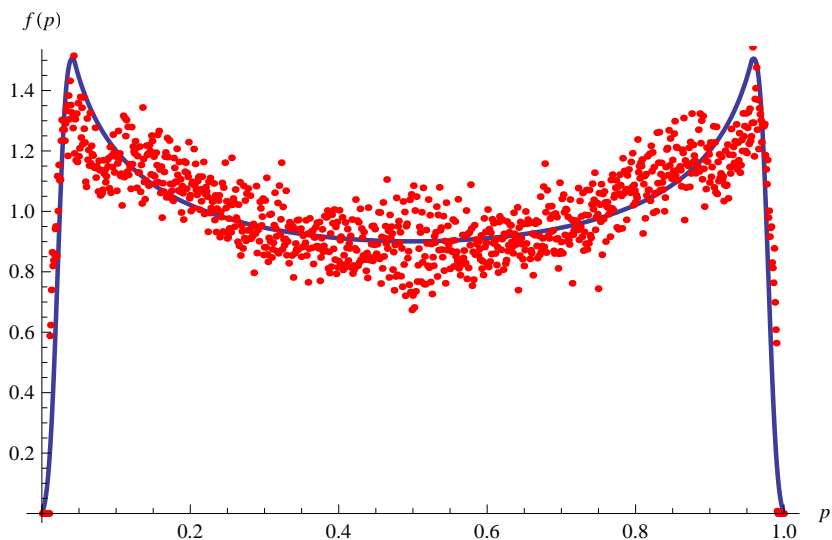}
   \end{tabular}
   \end{center}
   \caption{\label{ejem2} Empirical relative frequencies (dots) of allelic probabilities of the TGEN sample, with $1,237,567$ SNPs and PDF of truncated $\Beta(\alpha,\alpha,t)+NT(0,\delta,t)$ (continuous line) with $1225$ individuals, where $757$ are cases and  $468$ controls. The empirical relative frequencies were estimated grouping the data in $1000$ intervals between $\left[ 0,1\right]$.}
\end{figure}

The $\AL_t$ model fits much better to a real dataset than the previous analyzed models. This is especially true for the tails of the probability distribution functions. Besides, the mean squared error (MSE) estimator decreases for the new model. As shown in Table~\ref{MSEs}, the minimum MSE is reached for the $\AL_t$ model with $70.16\%$ decrease in the MSE when compared with the uniform model.

\begin{table}[htbp]
	\centering
	\caption{\textbf{\label{MSEs} The Mean Squared Error calculated in the four different distributions described fitting the empirical allelic probability grouped in $1000$ intervals.}}
		\begin{tabular}{|r|c|c|c|c|}
			\hline
		  TGEN AP  & Uniform & Beta & Beta$_t$ & $\AL_t$ \\ \hline
		  $MSE$ & $0.0429$  & $0.0517$ & $0.0315$  & $0.0128$ \\ \hline
		\end{tabular}
		
\end{table}

\subsection{Case-control probabilities}\label{SectionCaseControlProbabilities}

We assumed a GWAS, where a SNP was typed for $N$ individuals, with $N_0$ controls and $N_1$ cases. $M_s$ denotes the number of individuals which present a certain genotype ($s$).

Next we attempted to test for the following null hypothesis,
\begin{equation}\label{test}
H_0: \text{ the absence of association in genotype }s.
\end{equation}

Notice that under this null, the probability of being case conditioned to having a given genotype $s$, $b_s$, is the same as the probability of being case, $b$.
In order to test this null, the following binomial distribution, $\textbf{Bin}(M_s,b_{s})$, consisting of number of cases with genotype $s$ can be considered.
Under the null, $b_s$ can be estimated as $\widehat{b}_{s}=\widehat{b}=N_{1}/N$. On the other hand, as we mentioned before that the most sensible distribution for $b_s$ is the beta distribution, $\Beta\left(\alpha_{M_s},\beta_{M_s}\right)$. This is based on the fact that the beta distribution is the conjugate prior of the binomial distribution (\cite{MacKay2003}).

The mean of $b_s$ can be estimated by $\mu=\hat{b_s}=\frac{N_1}{N}$ and its variance $\sigma^2=\frac{\hat{b_s}(1-\hat{b_s})}{M_s}=\frac{N_0N_1}{N^2M_s}$. Since the population under study is finite it is necessary to adjust the variance $\sigma^2$ for the population size with the finite population correction factor $\frac{N-M_s}{N-1}$ (\cite{Isserlis1918}). This is specially required when sample size $M_s$ is not small in comparison with the population size $N$, so that $M_s>0.05N$. Therefore, when $M_s>5\%N$, the variance $\sigma^2$ remains as
$$\sigma^2=\frac{N_0N_1(N-M_s)}{N^2M_s(N-1)}.$$

Thus, under the null hypothesis, we can estimate $\alpha_{M_s}$ and $\beta_{M_s}$,
    \begin{equation}\label{alphaM_betaM}
     \alpha_{M_s} = \frac{N_1}{N}\left( \frac{M_s(N-1)}{N-M_s}-1\right) \quad \text{and} \quad \beta_{M_s} = \frac{N_0}{N}\left( \frac{M_s(N-1)}{N-M_s}-1\right).
    \end{equation}

Figure~\ref{PDFBsMFijoH0} compares the empirical relative frequencies of $b_s$ in TGEN dataset (dots) with the PDF of $\Beta\left(\alpha_{M_s},\beta_{M_s}\right)$ (continuous line) i.e., the beta distribution of $b_s$ under the null hypothesis, for several values of $M_s$.

\begin{figure}[h!]
    \begin{center}
    \begin{tabular}{ccc}
	\includegraphics[height=3cm]{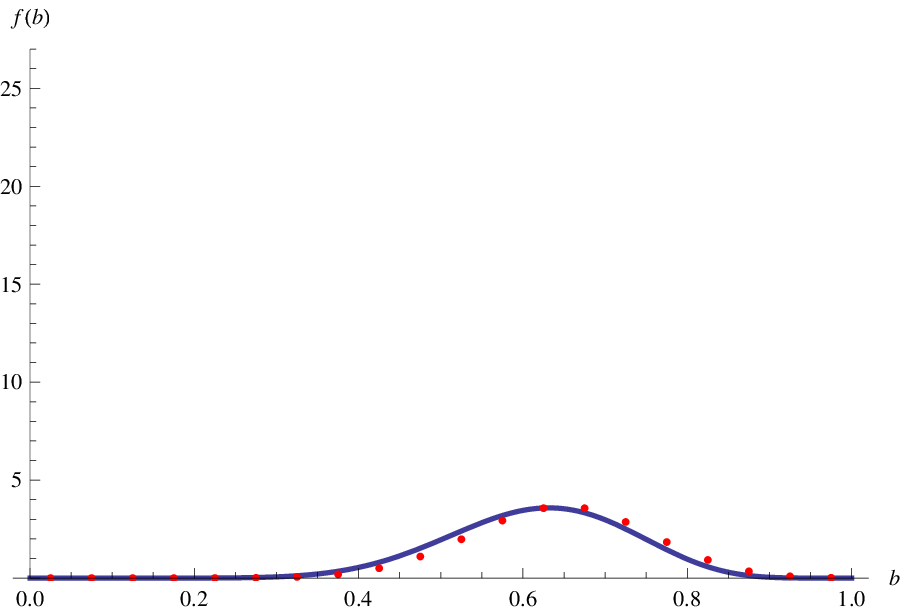} &
	\includegraphics[height=3cm]{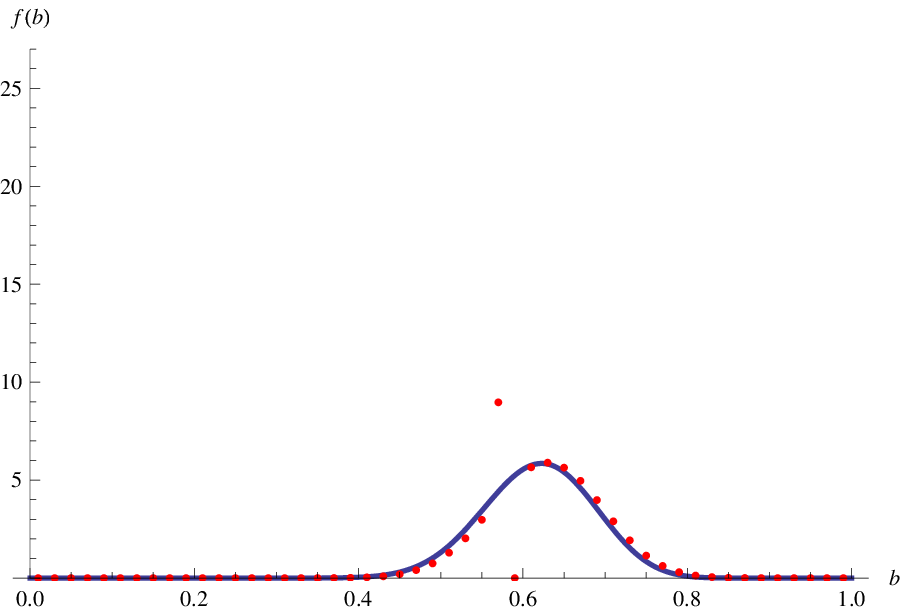} &
	\includegraphics[height=3cm]{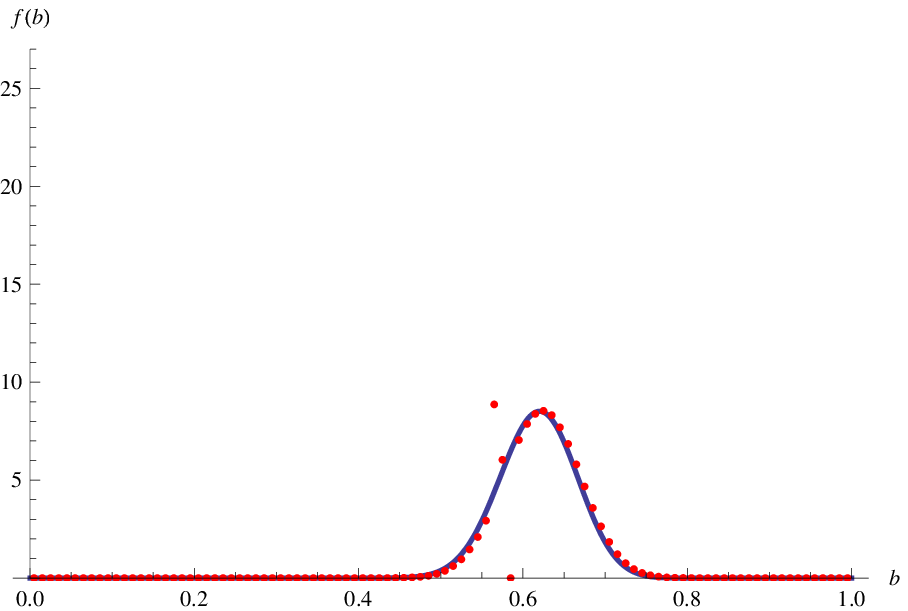} \\
	$M_s=20$ & $M_s=50$ & $M_s=100$ \\
	\includegraphics[height=3cm]{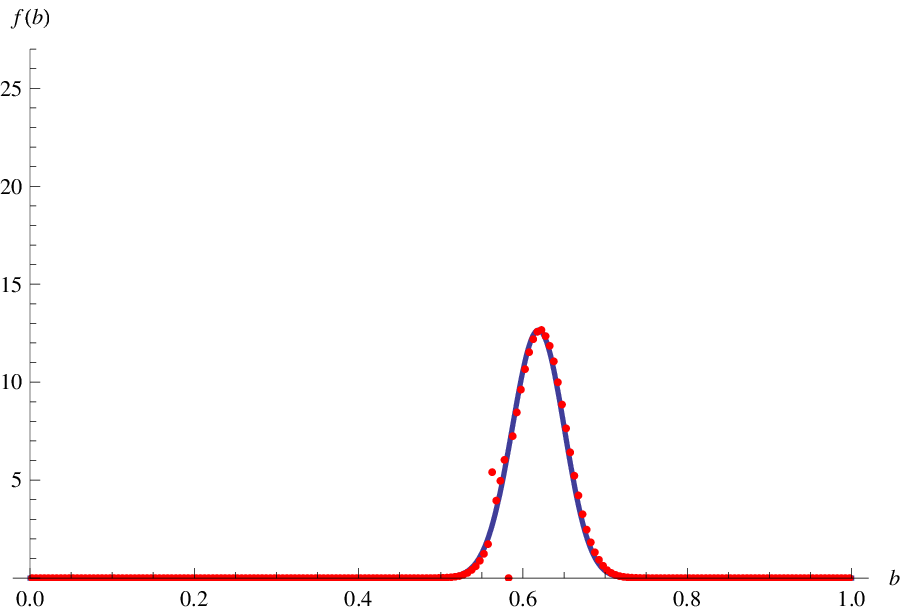} &
	\includegraphics[height=3cm]{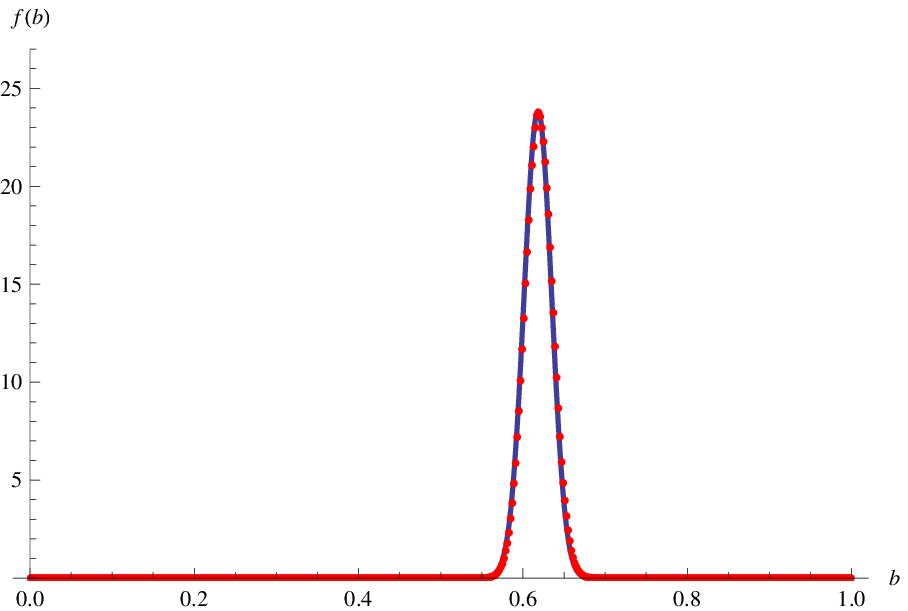} &
	\includegraphics[height=3cm]{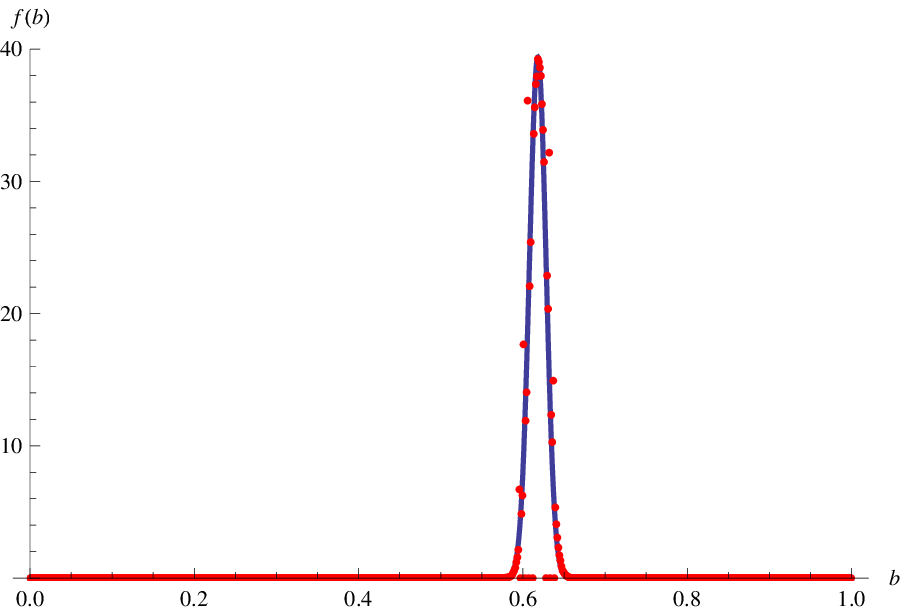} \\
	$M_s=200$ & $M_s=500$ &  $M_s=800$
    \end{tabular}
    \end{center}
\caption{\label{PDFBsMFijoH0} Empirical relative frequencies of $b_s$ (dots) for genotypes with different numbers of individuals (20, 50, 100, 200, 500 and 800 of individuals) and PDF of $\Beta(\alpha_{M_s},\beta_{M_s})$ (continuous line).}
\end{figure}

Thus, we obtained a decision rule for $H_0$ at a desired confidence interval (CI) as,

\begin{equation}\label{testbetacool}
\begin{array}{lcl}
\mbox{ Accept }H_0&\mbox{ if }& q_{\epsilon/2}\leq b_{s}\leq q_{1-\epsilon/2}\\
&&\\
\mbox{ Reject }H_0& &\mbox{ otherwise }
\end{array}
\end{equation}
 where $\epsilon$ is the type I error, $q_{\epsilon/2}$ and $q_{1-\epsilon/2}$ the extreme p-values, and both are tied by the beta distribution as,
\begin{equation}\label{testbetaCI}
	\epsilon/2 =Pr(\Beta(\alpha_{M_s}, \beta_{M_s})<q_{\epsilon/2})=Pr(\Beta(\alpha_{M_s}, \beta_{M_s})>q_{1-\epsilon/2}).
\end{equation}

Let $p$ denote the proportion of individuals of the general population which presents a given phenotype, and let $c_s$ denote the same proportion calculated in the individuals of a sample with a given genotype $s$. The ratio $\varphi_s=c_s/p$ is called the effect of the genotype $s$ on the phenotype in a given sample and represents the proportion in which an individual with genotype $s$ has more probability of presenting a phenotype than general population.
However, controls are a representation of the population, while cases are a sample of individuals with the phenotype. Therefore, it is possible to use the GWAS to find and estimation for $c_s$, (see Appendix \ref{understanding_case-control}). Also, $\varphi_s$ and $b_s$ can be related as
\begin{equation}
	\varphi_s = \frac{N_{0}}{N_{1}} \frac{b_{s}}{1-b_{s}}, \quad
b_{s} = \frac{\varphi_s N_{1}}{N_{0} + \varphi_s N_{1}}.
\end{equation}

 These last expressions allow constructing a decision rule that takes into account the effect ($\varphi_s$) of the genotype $s$.

 The new decision rule is written in the same way as \eqref{testbetacool} but in this case,
\begin{equation}\label{new-testbeta}
	\epsilon/2 =Pr(\Beta(\xi_{0}, \xi_{1})<q'_{\epsilon/2})=Pr(\Beta(\xi_{0}, \xi_{1})>q'_{1-\epsilon/2}),
\end{equation}
where,
\begin{equation}\label{new-testbetaPAR}
	\xi_{0} = \frac{\varphi_s N_{1}}{N_{0} + \varphi_s N_{1}}\left( \frac{M_s(N-1)}{N-M_s}-1\right) \text{ and } \xi_{1} = \frac{N_{0}}{N_{0} + \varphi_s N_{1}}\left( \frac{M_s(N-1)}{N-M_s}-1\right).
\end{equation}

The measured effect, applying the decision rule (\ref{testbetacool}) ($\varphi_s$), takes only into account the data contained in the sample. However, the effect ($\varphi_s$) that delimits the rejection region for a given confidence level $\epsilon$ can be computed from decision rule (\ref{new-testbeta}). Here, we called this effect as the \emph{critical effect} of genotype $s$, with a certain confidence level $\epsilon$ and it is denoted as $\Psi_s^{\epsilon}$ (Appendix~\ref{understanding_case-control}).

\subsubsection{Exploring BT in different genetic models}
The test for GWAS explained above, hereafter, would be referred as the Beta Test (BT) that considers each genotype characteristic independently.
This approach provides some new possibilities with respect to the conventional association test based on contingency tables and hereafter, would be called as  Conventional Test (CT).
In univariate CT, one can consider five genotype characteristics corresponding to the two alleles ($A$ and $B$) and their three possible combinations ($AA$, $AB$ and $BB$), separately.

Many contingency tables can be constructed in the CT by considering two or three genotype characteristics.
Similarly, many models can be used for the BT, considering the different non-empty subsets of the genotype characteristic. For example, there are eight possibilities for both univariate CT and BT.
However, most of the models are not meaningful for genetic purposes. For example, the allelic feature against the homozygous genotype or the heterozygous presence compared to the dominant model are two comparisons that do not provide useful information. On the other hand, the odds-ratio is not well defined when more than two characteristics are considered.

Therefore, in practice, only five models are used, namely \textit{allelic} ($A$ versus $B$), \textit{dominant} ($AA\cup AB$ versus $BB$), \textit{recessive} ($AA$ versus $AB\cup BB$), \textit{heterozygous} ($AA\cup BB$ versus $AB$) and \textit{genotype} ($AA$ versus $AB$ versus $BB$).

Next, we focused our attention in the allelic CT model and the alleles A and B of BT models because of the fact that the allelic CT model is the most common model used in GWAS. The relationships among the models can be seen in Appendix~\ref{Moritz}. The remaining models are available from the authors upon request.

\subsection{Measuring (or relating) locus effect in BT models}
Conventionally, odds-ratio is the effect measure estimated in a case-control GWAS. Let $OR_a$ be the odds-ratio in allelic model between cases and controls.
The magnitude of effect of the allelic CT model ($OR_a$) can also be related to the effect of the BT. In other words, the effect of a given model can be expressed as a function of $OR_a$.

In order to simulate a case-control population, it is necessary to fix an odds-ratio value, for example, $OR_a$, and another parameter like the MAF of a population of controls. Let $x_0$ be the MAF in a population of controls, and $x_1$ the MAF in a population of cases. If $OR_a$ is fixed, $x_1$ can be written as
\begin{equation}\label{ORMAF}
x_1=\frac{OR_a\cdot x_0}{1-x_0\cdot\left(1-OR_a\right)}.
\end{equation}
 Next, using the Hardy-Weinberg equilibrium (for short, HWE), it is possible to simulate the case-control population of SNPs.

The odds-ratio, as usually understood, shows the difference between the probability distribution for being case and the probability distribution for being control.
The BT assumed a different null hypothesis: in each genotype characteristic, the probability of being case is contained in a confidence interval centered around the proportion between cases and total population $N_1/(N_0+N_1)$. This proportion can be measured with the effect $\varphi_s$, and it is fixed for each genotype characteristic.
From the HWE and using (\ref{ORMAF}), the $OR_a$ and $\varphi_s$ (for $s=A$ or $B$) are related as follows:
	\begin{equation}\label{OREF}
	\varphi_A=\frac{OR_a}{\left(1-x_0\cdot(1-OR_a)\right)} \quad \text{and} \quad	\varphi_B=\frac{1}{\left(1-x_0\cdot (1-OR_a)\right)}.
	\end{equation}
	
Remarkably, if allelic frequencies in cases and controls are equal, there is no effect in any model. In other words, $x_0=x_1$ implies that $OR_a=1$ and $\varphi_s=1$.
Figure~\ref{plotsEF} shows the effects ($\varphi_s$) for different models as a function of the allelic odds-ratio ($OR_a\in \left[0, 10\right]$).

\begin{figure}[h!]
    \begin{center}
    \begin{tabular}{cc}
\includegraphics[height=5cm]{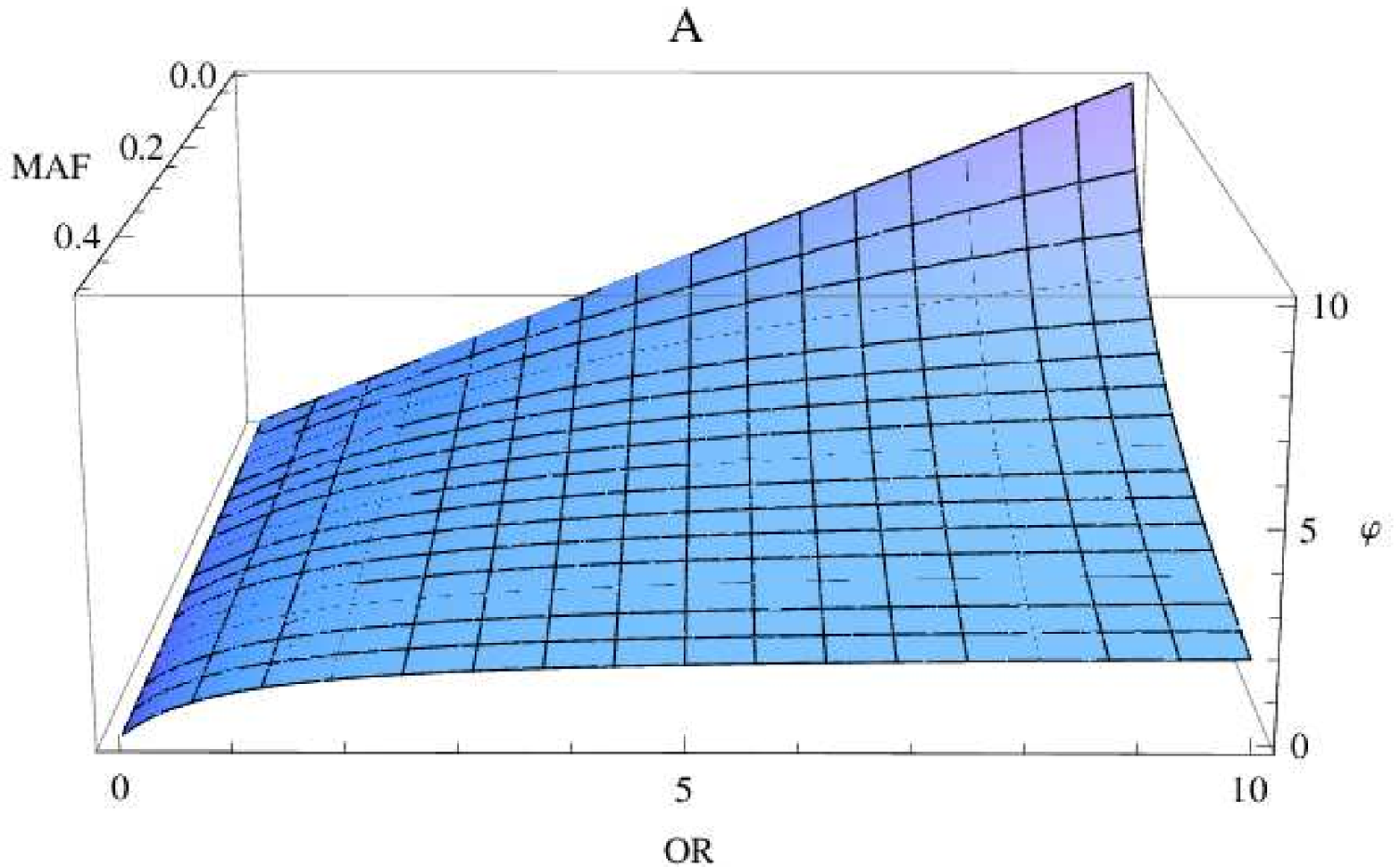} &
\includegraphics[height=5cm]{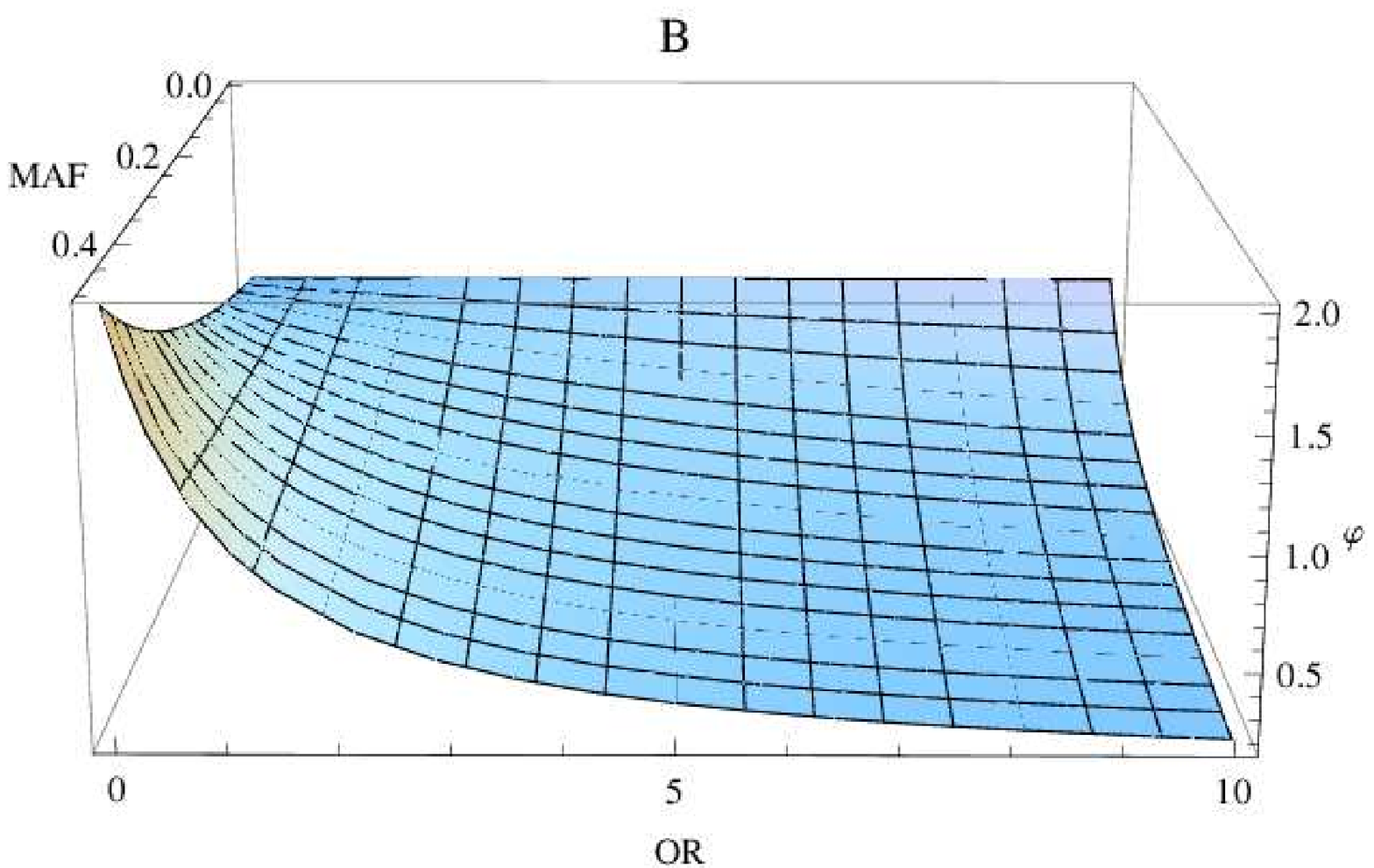} \\
    \end{tabular}
    \end{center}
\caption{\label{plotsEF} In the 3D graphic A, there are the effect of $s=A$ ($\varphi_A$) respect of the $OR_a$ between zero and $10$. In the 3D graphic B, there are the effect of $s=B$ ($\varphi_B$) respect of the $OR_a$ between zero and $10$, for different values of $MAF$.}
\end{figure}

\section{Monte Carlo Simulation}
A simulated dataset, consisting of $N_1=1000$ cases and $N_0=1000$ controls, was analyzed with both conventional chi-square (CT) and the new association (BT) tests in allelic models. A total of $10^5$ SNPs, with different minor allele frequencies ($0\leq\MAF\leq 0.5$), were simulated under the null hypothesis i.e., to have no effect on the trait ($OR_a=1$ or all the effects are equal to 1). This analysis can reflect whether the new Beta Test (BT) conforms to the theoretical beta distribution, and compares how the Conventional Test (CT) conforms to the Chi-square distribution. For estimating  the type I error for each test  (size of the test), we counted the number of test-statistics that had values above the critical values of the expected distribution.

\begin{table}[h!]
\centering
\caption{\textbf{\label{typeIsim} Type I Error (size) for each model of Beta and Conventional tests, respectively, taking the $\MAF$ as a random value between $[0,0.5]$ and different critical values of the expected distribution with corresponding confidence levels ($0.05$, $0.01$ and $0.001$).}$^1$}

\begin{tabular}{|c|c|c|c|c|c|c|c|c|}
\hline
\multirow{2}[4]{*}{\textbf{Expected}} & \multicolumn{8}{c|}{\textbf{BT}} \bigstrut\\
\cline{2-9}  & \textbf{Allele A} & \textbf{Allele B} & \textbf{DomA}  & \textbf{DomB} & \textbf{Homoz} & \textbf{s=AA} & \textbf{s=AB} & \textbf{s=BB} \bigstrut\\
\hline
\textbf{0.05} & 0.0489 & 0.0491 & 0.0493 & 0.0476 & 0.0491 & 0.0496 & 0.049 & 0.0494 \bigstrut\\
\hline
\textbf{0.01} & 0.00989 & 0.00989 & 0.00986 & 0.00926 & 0.0103 & 0.00996 & 0.0104 & 0.00982 \bigstrut\\
\hline
\textbf{0.001} & 0.000978 & 0.000957 & 0.00101 & 0.000739 & 0.000926 & 0.00107 & 0.000968 & 0.000989 \bigstrut\\
\hline
\end{tabular}%
\\
\begin{tabular}{|c|c|c|c|c|c|}
\hline
\multirow{2}[4]{*}{\textbf{Expected}} & \multicolumn{5}{c|}{\textbf{CT}} \bigstrut\\
\cline{2-6}  & \textbf{ALLELIC} & \textbf{DOM} & \textbf{REC} & \textbf{HETEROZ} & \textbf{GENO} \bigstrut\\
\hline
\textbf{0.05} & 0.0491 & 0.0494 & 0.0477 & 0.0492 & 0.0466 \bigstrut\\
\hline
\textbf{0.01} & 0.00989 & 0.00981 & 0.00926 & 0.0103 & 0.00937 \bigstrut\\
\hline
\textbf{0.001} & 0.000957 & 0.000978 & 0.000739 & 0.000916 & 0.000947 \bigstrut\\
\hline
\end{tabular}

\footnotesize{
\begin{flushleft} $^1$ The models are composed in BT by Allele $A$, Allele $B$, Dominant of $A$ ($AA\cup AB$), Dominant of $B$ ($AB\cup BB$), Homozygous ($AA\cup BB$)$s=AA$, $s=AB$ and $s=BB$. While in CT, the models are \textit{allelic} ($A$ versus $B$), \textit{dominant} ($AA\cup AB$ versus $BB$), \textit{recessive} ($AA$ versus $AB\cup BB$), \textit{heterozygous} ($AA\cup BB$ versus $AB$) and \textit{genotype} ($AA$ versus $AB$ versus $BB$). \end{flushleft}}
\end{table}

In order to simulate the case-control populations under certain conditions, only two parameters (MAF and an odds-ratio or $\varphi_s$) need to be fixed, as explained before. Indeed, fixing $\MAF$ in controls and the effect (odds-ratio or $\varphi_s$), the entire population is already defined.

The populations were simulated by generating at each SNP a random value for genotypic probabilities of controls according to $\MAF$ value and HWE. Once controls' $\MAF$ value and effect were fixed, the $\MAF$ values for cases were straightforwardly obtained, and following HWE, the genotype probabilities were computed, generating the cases population.

Using this method, both tests, CT and BT, yielded approximately the expected number of false positives (see Table~\ref{typeIsim}), suggesting they all conform to the expected theoretical distributions.

To estimate the power of both CT and BT, we carried out an analysis of a simulated dataset of 1000 cases and 1000 controls. Sets of $10^5$ SNPs were simulated under different alternative hypothesis (see Table~\ref{powersim}), with different effect sizes and minor allele frequencies ($0.05$, $0.2$ and $0.4$), with a fixed confidence level $\epsilon=0.05$. The effect of the genotype characteristics is related with the odds-ratio of corresponding model. Note that some scenarios could not be estimated due to HWE restrictions. Our results showed that both tests detected SNPs in association with similar power, depending on the new effect measurement and presenting there the allelic models comparison.

\begin{table}[h!]
\centering
\caption{\textbf{\label{powersim}  Power of both tests BT and CT in simulated data, taking SNPs under the alternative of association, with different minor allele frequencies ($\MAF$) and effect sizes at allelic models ($\varphi_A$ and $\varphi_B$).}$^1$}
\scalebox{0.7}{
\begin{tabular}{|c|c|c|c|c|c|c|c|}
\hline
\multicolumn{2}{|c|}{\textbf{Allele A}} & \multicolumn{2}{c|}{\boldmath{}\textbf{MAF = 0.05 $\varphi_A \in [0,20]$}\unboldmath{}} & \multicolumn{2}{c|}{\boldmath{}\textbf{MAF = 0.2 $\varphi_A \in [0,5]$}\unboldmath{}} & \multicolumn{2}{c|}{\boldmath{}\textbf{MAF = 0.4 $\varphi_A \in [0,2.5]$}\unboldmath{}} \bigstrut\\
\hline
\boldmath{}\textbf{Fixed $\varphi_A$}\unboldmath{} & \textbf{MODEL} & \boldmath{}\textbf{POWER (in $\%$)}\unboldmath{} & \boldmath{}\textbf{$OR_a$}\unboldmath{} & \boldmath{}\textbf{POWER (in $\%$)}\unboldmath{} & \boldmath{}\textbf{$OR_a$}\unboldmath{} & \boldmath{}\textbf{POWER (in $\%$)}\unboldmath{} & \boldmath{}\textbf{$OR_a$}\unboldmath{} \bigstrut\\
\hline
\multirow{2}[4]{*}{\textbf{0.5}} & \textbf{BT A} & 98.8 & \multirow{2}[4]{*}{0.49} & 100 & \multirow{2}[4]{*}{0.44} & 100 & \multirow{2}[4]{*}{0.38} \bigstrut\\
\cline{2-3}\cline{5-5}\cline{7-7}  & \textbf{CT ALLELIC} & 98.89 &   & 100 &   & 100 &  \bigstrut\\
\hline
\multirow{2}[4]{*}{\textbf{0.67}} & \textbf{BT A} & 74.4 & \multirow{2}[4]{*}{0.66} & 99.99 & \multirow{2}[4]{*}{0.62} & 100 & \multirow{2}[4]{*}{0.55} \bigstrut\\
\cline{2-3}\cline{5-5}\cline{7-7}  & \textbf{CT ALLELIC} & 75.65 &   & 99.99 &   & 100 &  \bigstrut\\
\hline
\multirow{2}[4]{*}{\textbf{0.8}} & \textbf{BT A} & 31.79 & \multirow{2}[4]{*}{0.79} & 91.12 & \multirow{2}[4]{*}{0.76} & 99.96 & \multirow{2}[4]{*}{0.71} \bigstrut\\
\cline{2-3}\cline{5-5}\cline{7-7}  & \textbf{CT ALLELIC} & 33.27 &   & 91.12 &   & 99.96 &  \bigstrut\\
\hline
\multirow{2}[4]{*}{\textbf{1.25}} & \textbf{BT A} & 40.26 & \multirow{2}[4]{*}{1.27} & 96.7 & \multirow{2}[4]{*}{1.33} & 100 & \multirow{2}[4]{*}{1.5} \bigstrut\\
\cline{2-3}\cline{5-5}\cline{7-7}  & \textbf{CT ALLELIC} & 40.21 &   & 96.7 &   & 100 &  \bigstrut\\
\hline
\multirow{2}[4]{*}{\textbf{1.5}} & \textbf{BT A} & 90.82 & \multirow{2}[4]{*}{1.54} & 100 & \multirow{2}[4]{*}{1.71} & 100 & \multirow{2}[4]{*}{2.25} \bigstrut\\
\cline{2-3}\cline{5-5}\cline{7-7}  & \textbf{CT ALLELIC} & 90.75 &   & 100 &   & 100 &  \bigstrut\\
\hline
\multirow{2}[4]{*}{\textbf{2}} & \textbf{BT A} & 100 & \multirow{2}[4]{*}{2.11} & 100 & \multirow{2}[4]{*}{2.67} & 100 & \multirow{2}[4]{*}{6} \bigstrut\\
\cline{2-3}\cline{5-5}\cline{7-7}  & \textbf{CT ALLELIC} & 100 &   & 100 &   & 100 &  \bigstrut\\
\hline
\hline
\multicolumn{2}{|c|}{\textbf{Allele B}} & \multicolumn{2}{c|}{\boldmath{}\textbf{MAF = 0.05 $\varphi_B \in [0,1.05]$}\unboldmath{}} & \multicolumn{2}{c|}{\boldmath{}\textbf{MAF = 0.2 $\varphi_B \in [0, 1.25]$}\unboldmath{}} & \multicolumn{2}{c|}{\boldmath{}\textbf{MAF = 0.4 $\varphi_B \in [0, 1.67]$}\unboldmath{}} \bigstrut\\
\hline
\boldmath{}\textbf{Fixed $\varphi_B$}\unboldmath{} & \textbf{MODEL} & \boldmath{}\textbf{POWER (in $\%$)}\unboldmath{} & \boldmath{}\textbf{$OR_a$}\unboldmath{} & \boldmath{}\textbf{POWER (in $\%$)}\unboldmath{} & \boldmath{}\textbf{$OR_a$}\unboldmath{} & \boldmath{}\textbf{POWER (in $\%$)}\unboldmath{} & \boldmath{}\textbf{$OR_a$}\unboldmath{} \bigstrut\\
\hline
\multirow{2}[4]{*}{\textbf{0.8}} & \textbf{BT B} & 100   & \multirow{2}[4]{*}{6} & 100   & \multirow{2}[4]{*}{2.25} & 100   & \multirow{2}[4]{*}{1.63} \bigstrut\\
\cline{2-3}\cline{5-5}\cline{7-7}      & \textbf{CT ALLELIC} & 100   &       & 100   &       & 100   &  \bigstrut\\
\hline
\multirow{2}[4]{*}{\textbf{0.9}} & \textbf{BT B} & 100   & \multirow{2}[4]{*}{3.22} & 100   & \multirow{2}[4]{*}{1.56} & 96.97 & \multirow{2}[4]{*}{1.28} \bigstrut\\
\cline{2-3}\cline{5-5}\cline{7-7}      & \textbf{CT ALLELIC} & 100   &       & 100   &       & 96.97 &  \bigstrut\\
\hline
\multirow{2}[4]{*}{\textbf{0.95}} & \textbf{BT B} & 100   & \multirow{2}[4]{*}{2.05} & 86.36 & \multirow{2}[4]{*}{1.26} & 48.67 & \multirow{2}[4]{*}{1.13} \bigstrut\\
\cline{2-3}\cline{5-5}\cline{7-7}      & \textbf{CT ALLELIC} & 100   &       & 86.36 &       & 48.7  &  \bigstrut\\
\hline
\multirow{2}[4]{*}{\textbf{1.05}} & \textbf{BT B} & 100   & \multirow{2}[4]{*}{0.048} & 90.96 & \multirow{2}[4]{*}{0.76} & 49.51 & \multirow{2}[4]{*}{0.88} \bigstrut\\
\cline{2-3}\cline{5-5}\cline{7-7}      & \textbf{CT ALLELIC} & 100   &       & 90.96 &       & 49.53 &  \bigstrut\\
\hline
\multirow{2}[4]{*}{\textbf{1.1}} & \textbf{BT B} & NA    & \multirow{2}[4]{*}{doesn't exist} & 100   & \multirow{2}[4]{*}{0.55} & 97.62 & \multirow{2}[4]{*}{0.77} \bigstrut\\
\cline{2-3}\cline{5-5}\cline{7-7}      & \textbf{CT ALLELIC} & NA    &       & 100   &       & 97.63 &  \bigstrut\\
\hline
\multirow{2}[4]{*}{\textbf{1.2}} & \textbf{BT B} & NA    & \multirow{2}[4]{*}{doesn't exist} & 100   & \multirow{2}[4]{*}{0.17} & 100   & \multirow{2}[4]{*}{0.58} \bigstrut\\
\cline{2-3}\cline{5-5}\cline{7-7}      & \textbf{CT ALLELIC} & NA    &       & 100   &       & 100   &  \bigstrut\\
\hline
\end{tabular}%
}

\footnotesize{ $^1$ First table is for allele A model and the second table for allele B model. For each effect size $\varphi$ and each minor allele frequency $\MAF$ the odds ratio of the allelic model $OR_a$ is reported using (\ref{OREF}), at a fixed confidence level $\epsilon=0.05$.}
\end{table}

\section{Application to Real Data}
In order to apply the BT approach to real data, a GWAS dataset with $1,237,567$ SNPs (where $N_0=468$ controls and $N_1=757$ cases) was used (Section~\ref{SectionRealData}).

Table~\ref{asTGEN} describes the parameters of the allelic probability distribution for cases, controls, and total population. As already shown in Figure~\ref{ejem2}, there are no allele frequencies lower than $0.01$; however, there exists some noise close to the extreme values.

\begin{table}[h!]
\centering
\caption{\textbf{\label{asTGEN} Summary of the TGEN allelic probabilities sample for cases, controls, and cases and controls.}}
\begin{tabular}{|c|}
\hline
TGEN\_impQC2: \\ 1,237,567 SNPs, 1225 individuals, 757 cases and 468 controls.\\
\hline
\begin{tabular}{r|ccc}
              & Cases \& Controls &  Cases &  Controls  \\ \hline
Empirical mean $\widehat{\mu}$  &    0.500065   &  0.500063       &      0.500065         \\
Empirical variance $\widehat{v}$  &  0.0916691     &   0.0917528      &    0.0917662     \\
Minimum $\widehat{v}$  &  0.01     &   0.01      &    0.01     \\
\textbf{Truncation $P_1$}  &    \textbf{ 0.018605 }   &     \textbf{0.018092 }      &   \textbf{  0.018233  }   \\
\textbf{$\widehat{\alpha}$}  &   \textbf{0.666681 }  &   \textbf{ 0.671064 }    &      \textbf{ 0.669102  }         \\	
\textbf{$\widehat{\delta}$}  &  \textbf{0.01871}         &  \textbf{0.01561}       &   \textbf{0.01623}      \\
Mean Squared Error &   0.0128   &      0.0154     &      0.013      \\ \hline
\end{tabular}
\end{tabular}

\end{table}

The calculated probability of being case with a genotype characteristic $s$ is very sensitive to parameter $M_s$. This is shown in Figure~\ref{PDFBsMFijoH0}, where $b_s$ is plotted for different values of $M_s$. This implies that results depend on the frequency of having the genotype $s$ and the sample size.
While making inference on the correlation of a phenotype and a genotype, one should take into account the number of individuals with such a genotype in the GWAS.

\begin{figure}[hb!]
    \begin{center}
        \begin{tabular}{cc}
         \includegraphics[height=8cm]{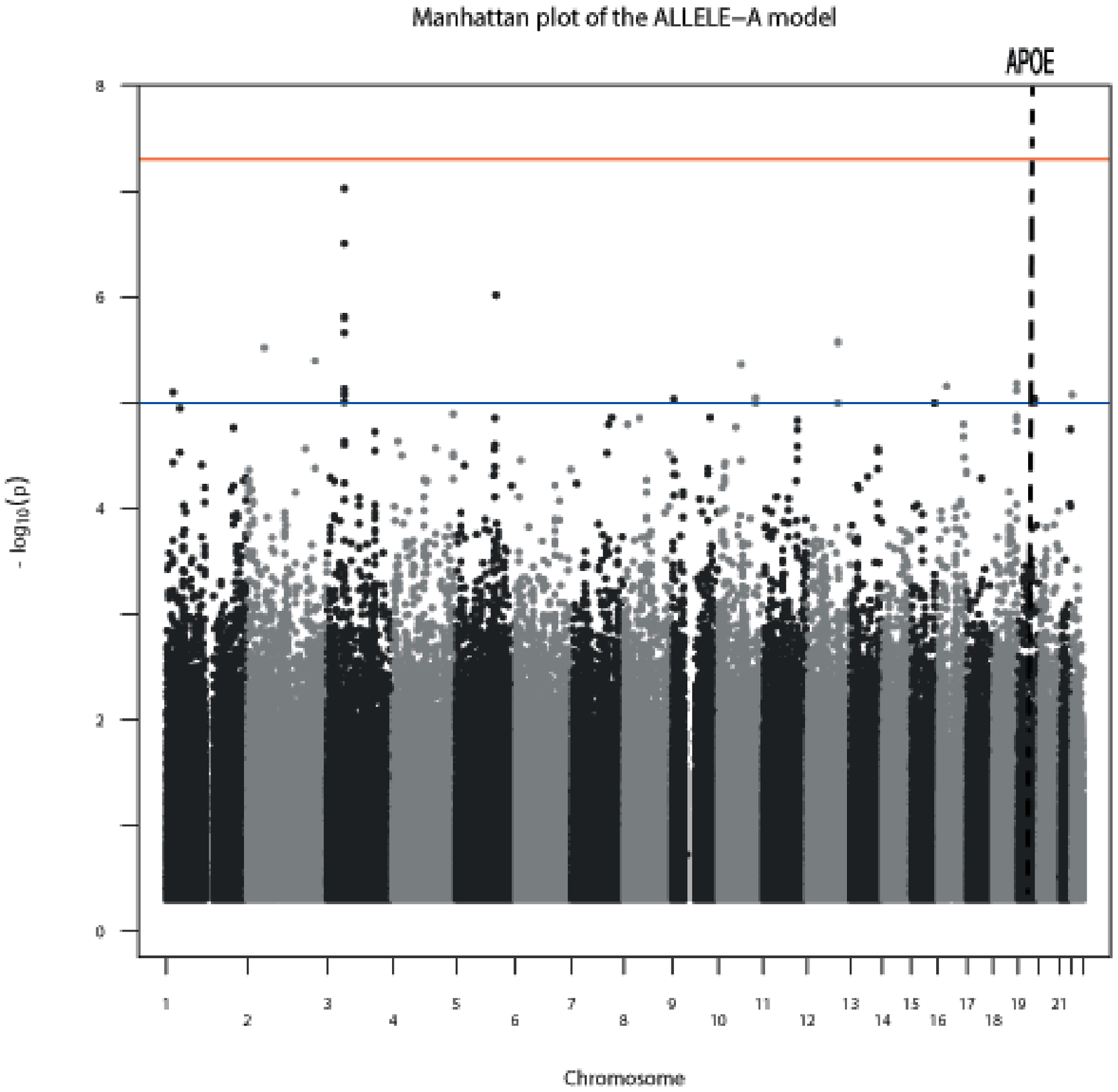} &
         \includegraphics[height=8cm]{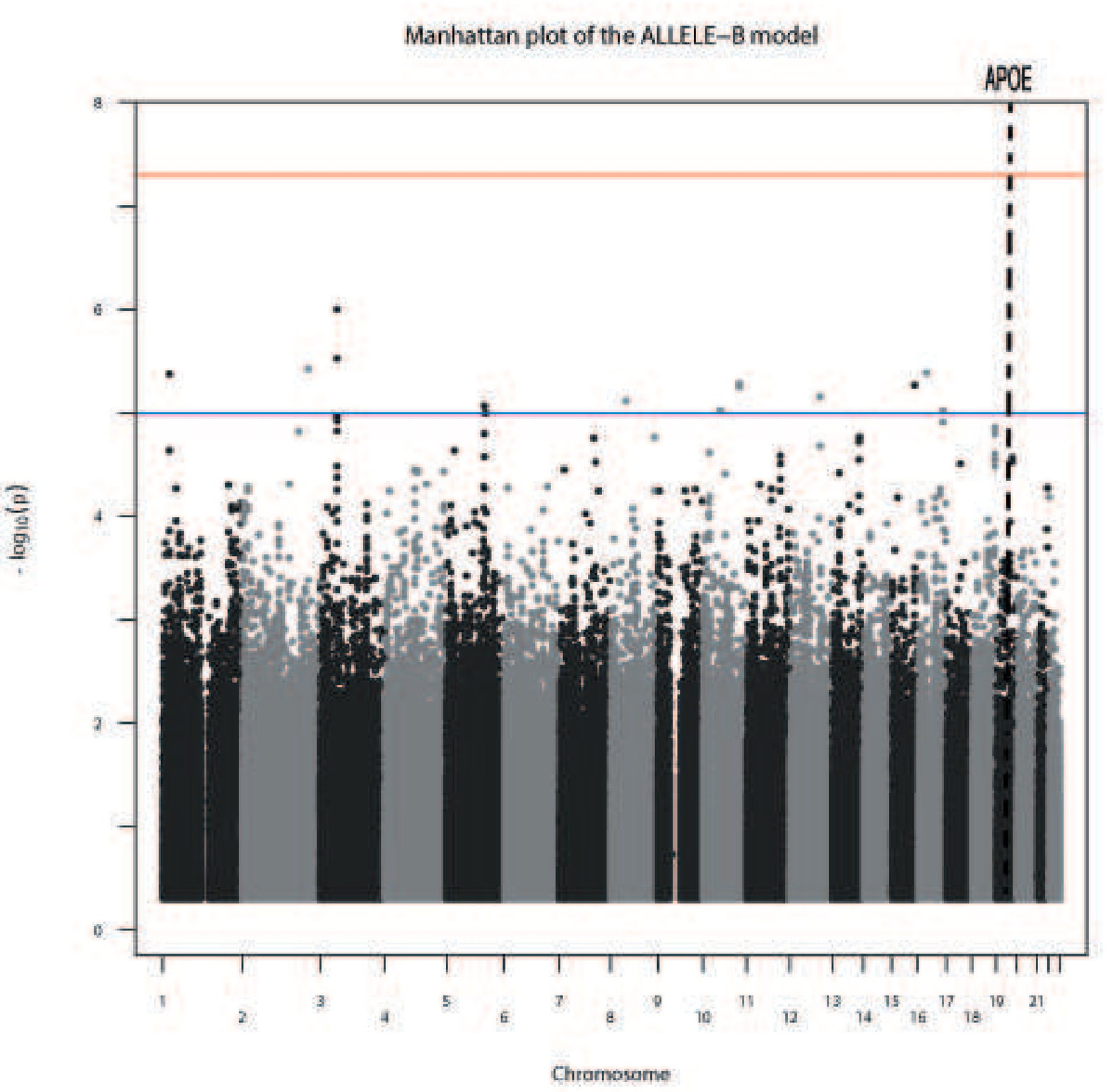}\\
	    \end{tabular}
    \end{center}
\caption{\label{Manhattan-plot} Manhattan plot of the new GWAS proposed with both alleles model in BT.}
\end{figure}

Figure~\ref{Manhattan-plot} shows the Manhattan plot applied to the BT in the Allele $A$ model within the sample described above (TGEN). Associated SNP on chromosome 19 is not displayed in the figure because of the chosen scale. This SNP is related to the apolipoprotein E (APOE) with a significant $p$-value$=1.33\cdot 10^{-42}$ on BT.

Note that the CT is a one-tailed test where the null hypothesis of no association is rejected if $p-$value is lower than $\epsilon$ and the direction of the effect determines the risk or protective role of the SNP. However the BT is a two-tailed test, where the null hypothesis of no association is rejected if $p$-value is lower than $\epsilon/2$ or bigger than $1-\epsilon/2$. In this case, the risk or protective role of the SNP was determined not only by the direction of the effect but also by the region of rejection. In other words, $p$-value$\leq \epsilon/2$ implies $\varphi_s \geq 1$ and \emph{vice-versa}.
Therefore, the $p$-value information is enough to define the risk or protection association.

\subsection{Measuring the concordance between CT and BT}

The agreement between two sets of results could be measured by the Kappa index ($K$) of agreement (\cite{Cohen1960}). In this case, we used two categories: being in association or not, with different levels of confidence. Kappa index is the estimator of agreement, compared in the paired models. At perfect agreement, $K$ equals to one, while agreement given by chance gives a value of $K$ close to zero.

Table~\ref{kappa} presents the allelic model comparison with their corresponding BT models (Allele A and Allele B). Notice that BT null hypothesis rejecting region contains approximately the expected number of significant SNPs than its corresponding confidence level. Nonetheless, there are some SNPs that reach significance in BT, but not in CT.

\begin{table}[h!]
\centering
\caption{\textbf{\label{kappa} The allelic paired models compared with the Kappa index of agreement, calculated in column $Kappa$, and the number of SNPs in concordance where both tests detect or not an association with the phenotype with a confidence level $\epsilon$. }$^1$}
\begin{tabular}{|r|c|c|c|c|c|c|}
\hline
\multicolumn{1}{|c|}{\textbf{K1}} &   & \multicolumn{5}{c|}{\textbf{ALLELE A}} \bigstrut\\
\hline
\multicolumn{1}{|c|}{} & cl & Kappa & Pos Conc & CT-BT+ & CT+BT- & Neg Conc \bigstrut\\
\hline
\multicolumn{1}{|c|}{\multirow{3}[6]{*}{\textbf{ALLELIC}}} & $\epsilon=0,05$ & 0.97 & 61861 & 1944 & 1651 & 1172111 \bigstrut\\
\cline{2-7}\multicolumn{1}{|c|}{} & $\epsilon=0,01$ & 0.94 & 11792 & 914 & 576 & 1224285 \bigstrut\\
\cline{2-7}\multicolumn{1}{|c|}{} & $\epsilon=0,001$ & 0.777 & 946 & 314 & 228 & 1236079 \bigstrut\\
\hline
\multicolumn{1}{|c|}{\textbf{K2}} &   & \multicolumn{5}{c|}{\textbf{ALLELE B}} \bigstrut\\
\hline
  & cl & Kappa & Pos Conc & CT-BT+ & CT+BT- & Neg Conc \bigstrut\\
\hline
\multicolumn{1}{|c|}{\multirow{3}[6]{*}{\textbf{ALLELIC}}} & $\epsilon=0,05$ & 0.993 & 63164 & 456 & 348 & 1173599 \bigstrut\\
\cline{2-7}\multicolumn{1}{|c|}{} & $\epsilon=0,01$ & 0.985 & 12263 & 266 & 105 & 1224933 \bigstrut\\
\cline{2-7}\multicolumn{1}{|c|}{} & $\epsilon=0,001$ & 0.976 & 1163 & 46 & 11 & 1236347 \bigstrut\\
\hline
\end{tabular}

\footnotesize{\begin{flushleft}
$^1$ \emph{Pos Conc} is the positive concordance, where both BT and CT reject the null hypothesis of no association. Similarly, \emph{Neg Conc} is the negative concordance, where both BT and CT accept the null hypothesis. \emph{BT$+$CT$-$} (respectively \emph{BT$-$CT$+$}) represents the number of SNPs in discordance on both tests, where BT refuse (respectively, accept) an association and CT accept (respectively, reject) it. \end{flushleft}}
\end{table}

\begin{table}[h!]
\centering
\caption{\textbf{\label{kappa2} Summary of the number of SNPs detected with association in CT Allelic model and both BT allele models ($A$ and $B$) and their differences, with a confidence level $\epsilon$.}$^1$ }
\begin{tabular}{|c|c|c|c|}
\hline
\multicolumn{4}{|c|}{Either allele A or B} \bigstrut\\
\hline
      & Pos Conc & CT- BT+ & CT+ BT- \bigstrut\\
\hline
$\epsilon=0.05$ & 63512 & 2273  & 0 \bigstrut\\
\hline
$\epsilon=0.01$ & 12368 & 1137  & 0 \bigstrut\\
\hline
$\epsilon=0.001$ & 1174  & 354   & 0 \bigstrut\\
\hline
\end{tabular}%

\footnotesize{
\begin{flushleft} $^1$  \emph{Pos Conc} is the positive concordance, where both BT and CT reject the null hypothesis of no association. \emph{BT$+$CT$-$} (respectively \emph{BT$-$CT$+$}) represents the number of SNPs in discordance on both tests, where BT refuse (respectively, accept) an association and CT accept (respectively, reject) it.\end{flushleft}}
\end{table}

When comparing CT with BT model, the Kappa index showed some differences in the association tests. For instance, a remarkable concordance ($K=0.976$, $\epsilon=0.001$) was observed while comparing the CT Allelic model with the BT Allele-B model (K2, Table~\ref{kappa}).

A remarkable result from Table~\ref{kappa} summarized in Table~\ref{kappa2}, reveals that all the SNPs, detected in association with the phenotype in CT test, were also detected in association with the BT test in any allele models ($A$ or $B$).
However, several SNPs detected being associated with the BT (in Allele $A$ or $B$ model) were not detected by the CT (see Table~\ref{kappa2}). For example, for a given confidence level $\epsilon=0.001$, the $23\%$ of positives, that are SNPs refusing $H_0$, were detected by BT allelic models since CT allelic model could not detect them.

The difference between BT results and the corresponding CT can be surprising at times. Indeed, one would expect much more concordance among CT and BT estimates. In order to understand the observed differences, a comparison of parameters space is advisable.

\subsection{Comparisons with Principal Component Analyses}

\begin{figure}[b!]
    \begin{center}
    \begin{tabular}{c}
	\includegraphics[height=15cm]{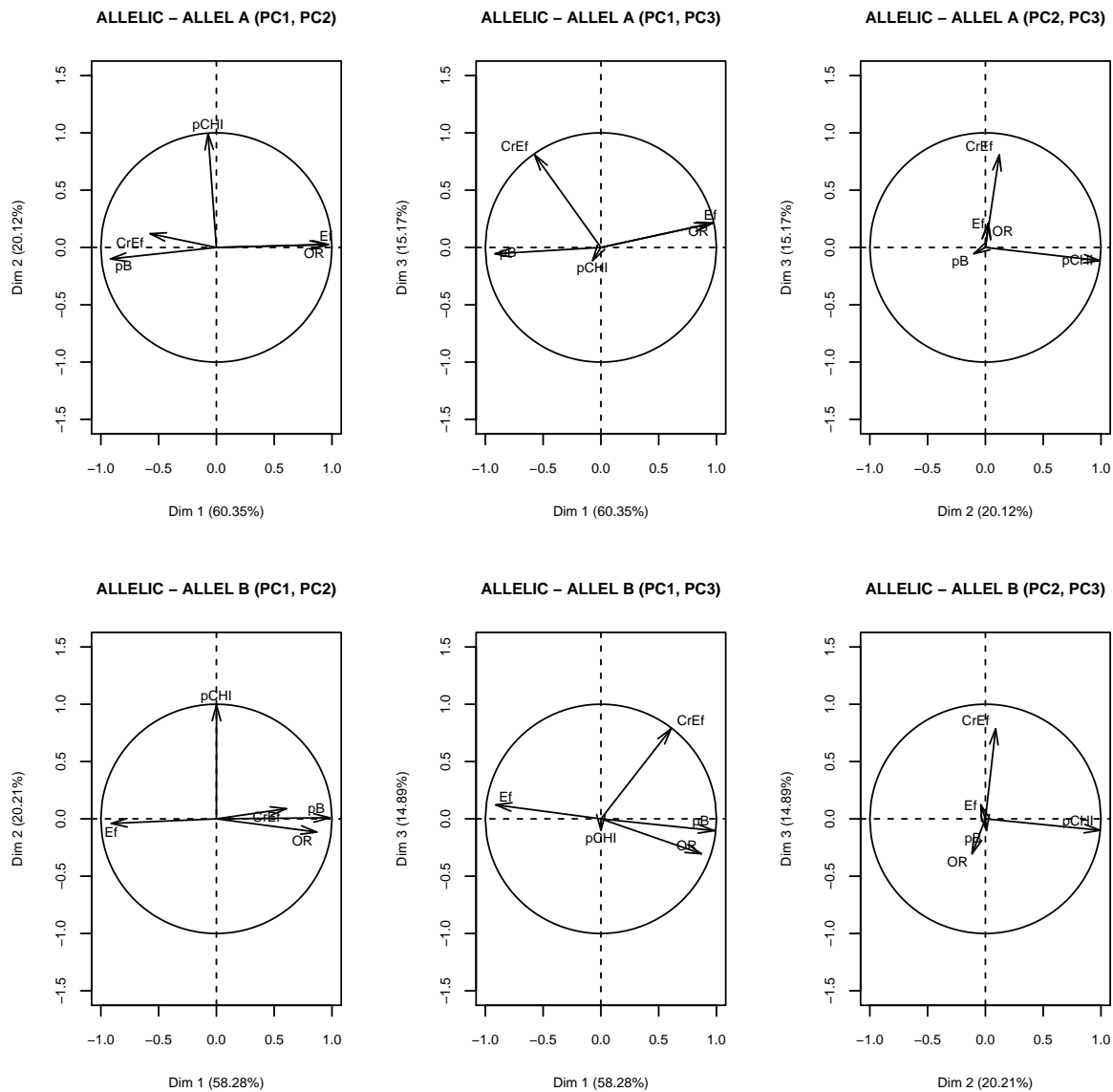}
	\end{tabular}
	    \end{center}
\caption{\label{pca} The weights of the parameters in each Principal Component (the first with the second, the first with the third and the second with the third, respectively), taken for the variables of the Allelic model of CT is compared to the Allele A and B models of BT. It describes the eigenvalues (in $\%$) of each principal component.
}
\end{figure}

Principal Components Analysis (PCA) is a measurement that shows the relationship between two sets of parameters. We performed here PCA for analyzing the relationship between the parameters from CT ($p$-value of the $\chi^2$ distribution and the odds-ratio $OR_i$ corresponding to the model) and from BT ($p$-value of the beta distribution, the  effect $\varphi_s$ and the critical effect $\Psi_s^\epsilon$). PCA finds which parameters explain the maximum variability and also sorts the components (transformed variables) by their explained variance; the original variables have corresponding weights in each components (\cite{Pearson1901}).

In summary, PCA found which variables explain the maximum variability in the CT and BT results, the percentage of explained variance and the intensity.


%

Different variability explanation between allelic CT model and both alleles BT models using three principal components is shown in Figure~\ref{pca} (a. and b.). These figures show the proportion explained variance is contributed by each parameter in the principal components. The $p$-value and $\varphi_s$ of the BT, and the odds-ratio $OR$ of the CT explain the same variance, which means these three parameters have the same direction as PC1.
As expected, the $p$-value of CT (called as $pCHI$) explains the variance along the second principal component. However, the \emph{Critical effect} of the BT $\Psi_s^\epsilon$ complements the third principal component, which means that some variability cannot be explained without it. As a matter of fact, the Critical effect contributes to the $87.72\%$ of the third principal component, which explains the $15.17\%$ of the total variance.


Therefore, the principal component analysis reflects that the results of both methods must differ in a region of the explaining variance.

\subsection{Ranking of association}

\begin{table}[h!]
\centering
\caption{\textbf{\label{ranking20} Top smallest $p-$values computed with the BT for the Alleles A and B models ($<10^{-5}$) in the TGEN dataset, described above (* symbol represents that several markers in LD have been omited in a single signal with the lowest $p-$value).}$^1$ }
\scalebox{0.65}{
\begin{tabular}{cccccccccccccccc}
CHR   & BP    & SNP   & MAF   & A1    & A2    & \textbf{pvBT} & BT-Allele & $\varphi_A$ & $\varphi_B$ & $Psi_A^{0.05}$ & $Psi_B^{0.05}$ & \textbf{pvCT} & $OR_a$  & $M_A$ & $M_B$ \\
19    & 50114786 & rs4420638 & 0,308 & G     & A  & \textbf{1,33E-42} & 1 & 2,44  & 0,72  & 2,12  & 0,76  & \textbf{2,21E-34} & 3,39598 & 754   & 1694 \\
3     & 52481466* & rs6784615 & 0,069 & C     & T  & \textbf{9,43E-08} & 1 & 2,30  & 0,95  & 1,64  & 0,97  & \textbf{2,09E-06} & 2,43106 & 170   & 2278 \\
5     & 121942614 & rs11953981 & 0,058 & G     & A     & \textbf{9,60E-07} & 1 & 2,45  & 0,95  & 1,70  & 1,03  & \textbf{1,83E-05} & 2,56819 & 128   & 2086 \\
12    & 93847903 & rs249152* & 0,193 & A     & G     & \textbf{2,64E-06} & 1 & 1,49  & 0,91  & 1,26  & 0,96  & \textbf{1,16E-05} & 1,63099 & 464   & 1938 \\
2     & 51804001 & rs17864593* & 0,016 & A     & G     & \textbf{3,02E-06} & 1 & 5,40  & 0,98  & 2,24  & 1,00  & \textbf{NA} & 5,51085 & 39    & 2329 \\
2     & 205265992 & rs41511746 & 0,017 & G     & C    & \textbf{3,76E-06} & 2 & 0,25  & 1,02  & 0,47  & 1,01  & \textbf{7,82E-06} & 0,241278 & 42    & 2408 \\
16    & 26556972 & rs12162084 & 0,157 & A     & G    & \textbf{4,10E-06} & 2 & 0,66  & 1,08  & 0,80  & 1,04  & \textbf{9,34E-06} & 0,611134 & 385   & 2063 \\
1     & 21773864 & rs1536934 & 0,069 & A     & G   & \textbf{4,21E-06} & 2 & 0,51  & 1,05  & 0,69  & 1,02  & \textbf{9,17E-06} & 0,479325 & 153   & 2057 \\
10    & 68271216 & rs4486514 & 0,04  & C     & T     & \textbf{4,30E-06} & 1 & 2,71  & 0,97  & 1,69  & 1,01  & \textbf{7,54E-05} & 2,80962 & 93    & 2231 \\
10    & 112534645 & rs7077757* & 0,211 & T     & C     &  \textbf{5,22E-06} & 2 & 0,71  & 1,10  & 0,83  & 1,05  & \textbf{1,21E-05} & 0,645986 & 515   & 1925 \\
15    & 90462008 & rs11074041* & 0,142 & C     & G   & \textbf{5,45E-06} & 2 & 0,64  & 1,08  & 0,78  & 1,03  & \textbf{1,23E-05} & 0,597382 & 335   & 2019 \\
18    & 71895842 & rs359739* & 0,192 & A     & C     & \textbf{6,64E-06} & 1 & 1,46  & 0,92  & 1,22  & 0,95  & \textbf{2,58E-05} & 1,59007 & 470   & 1976 \\
8     & 47534249 & rs4313171 & 0,067 & T     & C      & \textbf{7,68E-06} & 2 & 0,53  & 1,05  & 0,71  & 1,03  & \textbf{1,65E-05} & 0,502977 & 165   & 2285 \\
22    & 17002691 & rs12168275 & 0,038 & G     & C     & \textbf{8,38E-06} & 1 & 2,72  & 0,97  & 1,65  & 1,00  & \textbf{0,00013} & 2,80573 & 87    & 2215 \\
5     & 117744488 & rs6595122 & 0,282 & C     & A     &  \textbf{8,61E-06} & 2 & 0,75  & 1,12  & 0,86  & 1,06  & \textbf{2,06E-05} & 0,667943 & 637   & 1621 \\
19    & 50323656 & rs17643262* & 0,082 & A     & G     & \textbf{9,07E-06} & 1 & 1,88  & 0,95  & 1,37  & 0,97  & \textbf{5,69E-05} & 1,97991 & 192   & 2148 \\
9     & 5583190 & rs10815248 & 0,042 & A     & T     & \textbf{9,23E-06} & 1 & 2,49  & 0,97  & 1,58  & 0,99  & \textbf{0,00011} & 2,57753 & 101   & 2309 \\
10    & 53698470 & rs10824310 & 0,065 & T     & C     &  \textbf{9,42E-06} & 2 & 0,52  & 1,05  & 0,71  & 1,03  & \textbf{2,02E-05} & 0,500702 & 159   & 2291 \\
16    & 77974064 & rs7192960 & 0,129 & T     & C      & \textbf{9,57E-06} & 2 &  0,64  & 1,07  & 0,79  & 1,04  & \textbf{2,12E-05} & 0,59958 & 316   & 2134 \\
\end{tabular}
}

\footnotesize{
\begin{flushleft} $^1$ The columns represent from left to right the number of chromosome (Chr), the base pair position of the SNP (BP), the name of the SNP (SNP), the MAF value, code for minor allele (A1), code for the other allele (A2), the lowest $p-$value of BT allelic models (pvBT), the number of the allele model with lowest $p-$value of BT allelic models (BT-Allele), the effect of BT allelic models ($\varphi_A$ and $\varphi_B$), the Critical effect of BT allelic models ($Psi_A^{0.05}$ and $Psi_B^{0.05}$), the $p-$value of the CT (pvCT), the odds-ratio of this model ($OR_a$) and the number of individuals which present this alleles ($M_A$ and $M_B$). \end{flushleft}}
\end{table}

Top results for BT analysis are presented in Table~\ref{ranking20}. Briefly, we displayed the ranking of the smallest $p$-values in BT ($<10^{-5}$) including either A or B allele estimation and their corresponding results using conventional one degree of freedom chi-squared tests applied to MAF (CT). We also included the ranking order observed for each marker using both approaches (BT and CT). As expected, most SNPs, but not all of them, exhibited very similar $p-$values and ranking order. This result is fully compatible with the performance of the  global kappa index (Table~\ref{kappa}).

A whole and comprehensive ranking is also included (Supplementary file ``RankingTGENpvalueBT.csv'').  Of course, both strategies identified SNP marker rs4420638 on chromosome 19, located 14 kilobase pairs distal to the APOE epsilon variant as the major finding. This observation was previously reported by TGEN researchers (\cite{Coon2007}). APOE locus is the most important genetic risk factor for Alzheimer's disease reported to date (\cite{Corder1993}).  Notably, we found APOE locus significance more than eight orders of magnitude smaller using BT compared to CT (Table~\ref{ranking20}).

The rest of top markers also display smaller $p$-values using BT compared to CT calculations. This can be explained by the fact that the CT is a one-tailed test where the null hypothesis of no association is rejected if $p$-value turns out to be less than   $\epsilon$, while the BT is a two-tailed test (where the risk or protective role of the SNP is known by the $p$-value) and the rejection area for BT is one half that for CT ($\epsilon/2$).

Notice that the non-available $p$-value for the CT in the Chromosome 2 is due to the lack of data in a given cell. Anyway, this can be corrected by the Fisher test, which is not generally recommended due its computational cost.

\section{Discussion}
 Any description of allelic distributions in the genome must begin by constructing a model of allelic probabilities. However, this important point remains unaddressed in many scientific literature, at least to the best of our knowledge. Indeed, almost all simulations, made for testing a GWAS method, assume that allelic frequencies follow a uniform distribution. Here a model for allelic probability distribution is proposed and tested. In addition, we also improved the commonly used uniform distribution model.
	
The proposed alleles probability model, $\AL_t$, offers a common scenario for each data set, characterizing noise. The truncation when not known can be estimated using the empirical distribution of the allelic probabilities. Regardless of the truncation,  the remaining noise (quality control, stratification, insufficient population, etc...) is gathered with  the variable $\D_t$.

The model depends only on the population of cases, controls, and the number of occurrences of the feature in the sample.
Although the study was focused in a univariate analysis and $s$ can be taken as the genotype in a single SNP, note that $s$ is, in general, a vector and can represent any desired condition. For instance, $s$ can include, along the genotype, the sex, age or other information. Furthermore, this vector can also include more than one genotype whether considering their interactions (epistasis) or not.

 Although the examples used in this work described the allelic model, any other models could also be also analyzed without further modification of the method.
	
	The new genome-wide association method (BT) has some commonalities with the conventional one. However, BT offers a remarkable ranking variability that might represent genuine signals.
Novel candidate must be corroborated by intensive replication, multiple testing control, and meta-analysis using other datasets.  We are aware that if  BT can isolate novel loci, which generally missed while using traditional approaches, its application may help to uncover a fraction of the missing piece of heritability still pending for multiple complex traits.
Consequently, next step of our research would be the generation of well powered meta-analysis of BT rankings. The isolation of genome-wide significant signals, and ultimately, the replication in independent series may help to measure the utility of this novel GWAS approach.

\appendix{}
\section{Allelic probabilities variations}\label{Appendix-Variation}

The difference between truncated universal allelic probability, $\A_t$, and a priori truncated local allelic probability, $\AL_t$, can be described in a single expression, $\D_t=\A_t-\AL_t$ the divergence of local MAF from commercial chips truncation for small $\MAF$.
Since the expected values of $\A_t$ and $\AL_t$ should be equal, we assumed that the expected value of $\D_t$ as zero, and its seemed a plausible assumption that $\D_t$ follows a normal distribution. However, this would give rise to local allelic probabilities out of the interval $[0,1]$. To avoid this, we assumed that the values of $\A_t$ in the SNPs present in the chip belong to the interval $[t,1-t]$, and the observed allelic probability takes values in $[0,1]$. Therefore, the observed values of $\D_t$ belong to the interval $[0,1]$. Hence, we modeled $\D_t$ as a truncated normal distribution $NT(0,\delta,-t,t)$. Its probability density function is given by
\begin{equation}
h(x)=\left\{\matriz{{ll} \frac{\frac{1}{\sqrt{2\pi}}e^{\frac{x^2}{-2\delta^2}}}{{\int_{-t}^{t}\frac{1}{\sqrt{2\pi}}e^{\frac{x^2}{-2\delta^2}}}}, & \text{if } -t<x<t;\\&\\0, & \text{otherwise},} \right.
\end{equation}

The mean of $\D_t$ is $0$ and its variance is given by
$$\sigma^2(\D_t)=\sigma_D^2=\delta^2\left( 1-\frac{2\frac{t}{\delta} (\frac{t}{\delta})}{2\Phi\left(\frac{t}{\delta}\right)-1} \right)$$
where $\phi$ and $\Phi$ are the PDF and CDF of the standard normal distribution $N(0,1)$, respectively.
If we take $x=\frac{t}{\delta}$ then
    $$\sigma_D^2=t^2\left( \frac{1-\frac{2x \phi(x)}{2\Phi(x)-1}}{x^2} \right)=t^2\left( \frac{1}{x^2}-\frac{2 \phi(x)}{x(2\Phi(x)-1)} \right).$$

As $\A_t$ and $\D_t$ were assumed to be independent the density function of $$\AL_t=\A_t+\D_t$$ is the convolution of the probability density functions of $\A_t$ and $\D_t$:
\begin{equation}\label{convo}
(g*h)(z)=\int_{-\infty}^{+\infty}g(z-x)h(x) dx
\end{equation}

\section{Properties of truncated allelic distribution}\label{Appendix-Lema}
As $\A_t$ is distributed as $\Beta(\alpha,\alpha,t,1-t)$, the mean of $\A_t$ should be $\mu_{\alpha,t}=0.5$. This can be easily deduced from its PDF $g(a)$ (see (\ref{betru-ok})).

In order to calculate the variance of $\A_t$, it is convenient to use the incomplete beta function, which is defined as
    $$\B(t;\alpha,\beta)=\int_0^t a^{\alpha-1}(1-a)^{\beta-1}da.$$
This function is related with the beta function $\Beta(\alpha,\beta)$ using the regularized incomplete beta function:
    $$ I_t (\alpha,\beta)= \frac{\B(t;\alpha,\beta)}{\B(\alpha,\beta)} $$
They satisfy the following properties (\cite{Paris2010})
\begin{eqnarray}\label{BetaIncompleta}
\B(\alpha+1,\beta)&=&\frac{\alpha}{\alpha+\beta} \B(\alpha,\beta) \nonumber \\
I_t (\alpha,\beta)&=& 1- I_{1-t}(\beta,\alpha)  \\
I_t (\alpha+1,\beta)&=&I_t(\alpha,\beta) - \frac{t^\alpha (1-t)^\beta}{\alpha \B(\alpha,\beta) } \nonumber
\end{eqnarray}

By (\ref{MediaVarianzaalphabeta-ok}), the variance $\sigma^2_{\alpha}$ of $\Beta(\alpha,\alpha)$ is $\frac{1}{4(2\alpha+1)}$. Using (\ref{BetaIncompleta}) we calculated the variance of the truncated beta distribution:
\begin{eqnarray*}
 \sigma_{\alpha,t}^2 & = &  \int_t^{1-t} \frac{a^{\alpha+1}(1-a)^{\alpha-1}}{\int_t^{1-t} r^{\alpha-1}(1-r)^{\alpha-1}} dr - \mu_{\alpha,t}^2  \\
  & = & \frac{ \B(\alpha+2,\alpha) (I_{1-t}(\alpha+2,\alpha) - I_t (\alpha+2,\alpha))}{\B(\alpha,\alpha)(I_{1-t}(\alpha,\alpha) - I_t (\alpha,\alpha))} -\frac{1}{4} \\
 & = & \frac{ \alpha +1}{4\alpha + 2} \frac{\frac{t^\alpha (1-t)^\alpha (4t-2)}{(\alpha+1)\B(\alpha,\alpha)} + 1- 2 I_t (\alpha,\alpha)}{1 - 2 I_t (\alpha,\alpha)} -\frac{1}{4} \\
  & = & \frac{1}{(4\alpha + 2)\B(\alpha,\alpha)} \frac{t^\alpha (1-t)^\alpha (4t-2)}{1 - 2 I_t (\alpha,\alpha)} + \frac{ \alpha +1}{4\alpha + 2} -\frac{1}{4}
\end{eqnarray*}

Therefore, the variance of $\A_t$ is
\begin{equation}\label{sigma_u-at}
    \sigma_u^2 = \sigma_{\alpha,t}^2 = \frac{1}{(4\alpha + 2)\B(\alpha,\alpha)} \frac{t^{\alpha} (1-t)^{\alpha} (4t-2)}{1 - 2 I_t (\alpha,\alpha)} + \frac{ \alpha +1}{4\alpha + 2} -\frac{1}{4}.
\end{equation}
Next, we needed the following technical lemma.

\begin{lemma}\label{1/3}
Let $f(x)=\frac{1}{x^2}-\frac{2 \phi(x)}{x(2\Phi(x)-1)}$. Then $f(x)$ is decreasing for $x>0$, $f(x)<\frac{1}{3}$ for every $x>0$ and $\lim_{x\rightarrow 0} f(x) =\frac{1}{3}$.
\end{lemma}

\begin{proof}
 First, we need to prove that $f(x)<0$, if $x>0$. For that, we consider the function $g(x)=2\Phi(x)-1+\frac{6x\phi(x)}{x^2-3}$. Having in mind that $\phi'(x)=-x\phi(x)$ and $\Phi'(x)=\phi(x)$, we have
    $$g'(x)=2\phi(x)\left(1+3\frac{(1-x^2)(x^2-3)-2x^2}{(x^2-3)^2}\right) = -\frac{4x^4\phi(x)}{(x^2-3)^2}\le 0.$$
Thus, $g$ is strictly decreasing function in each region where it is continuous, for example, the interval $(0,\sqrt{3})$.
Let $0<x<\sqrt{3}$. As $g(0)=0$, we deduce that $g(x)<0$. As $x^2-3<0$, we have $0<(x^2-3)g(x)=(2\Phi(x)-1)(x^2-3)+6x\phi(x)$. Pleanly $(2\Phi(x)-1)(x^2-3)+6x\phi(x)$ is also positive if $x\ge \sqrt{3}$. This proves that $(2\Phi(x)-1)(x^2-3)+6x\phi(x)>0$ for every $x>0$. Using that $2\Phi(x)-1,x^2>0$, we deduce that $\frac{3-x^2}{3x^2}-\frac{2\phi(x)}{x(2\Phi(x)-1)}<0$ and hence
    $$f(x)=\frac{1}{x^2}-\frac{2 \phi(x)}{x(2\Phi(x)-1)} = \frac{1}{3}+\frac{3-x^2}{3x^2}-\frac{2\phi(x)}{x(2\Phi(x)-1)}<\frac{1}{3},$$
as desired.

Using L'H\^{o}pital rule, we have $\lim_{x\rightarrow 0}\frac{2\Phi(x)-1}{x\phi(x)} = \lim_{x\rightarrow 0}\frac{2}{1-x^2} = 2$. Applying L'H\^{o}pital again, we have
    \begin{eqnarray*}
    \lim_{x\rightarrow 0} \left(\frac{1}{x^2}-\frac{2 \phi(x)}{x(2\Phi(x)-1)}\right) &=& \lim_{x\rightarrow 0} \frac{2\Phi(x)-1-2x \phi(x)}{x^2(2\Phi(x)-1)} \\
    &=& \lim_{x\rightarrow 0} \frac{x\phi(x)}{2\Phi(x)-1+x\phi(x)}
    = \lim_{x\rightarrow 0} \frac{1}{\frac{2\Phi(x)-1}{x\phi(x)}+1} =\frac{1}{3}
    \end{eqnarray*}

To prove that $f(x)$ is decreasing for $x>0$, we consider the following functions
    \begin{eqnarray*}
    \psi(x) &=& 2\Phi(x)-1 \\
     \alpha(x) &=& (x^4-2x^2+3) \sqrt{x^4+2x^2+9}+x^6-x^4+5x^2-9. \\
    h(x) &=& 2\psi(x)-x\phi(x)(x^2+1+ \sqrt{x^4+2x^2+9}).
    \end{eqnarray*}

We claim that $\alpha(x)>0$ for every $x>0$. For that we use the following equality
    \begin{equation}
     (x^4-2x^2+3)^2 (x^4+2x^2+9)-(x^6-x^4+5x^2-9)^2=32x^4.
    \end{equation}

Therefore
$$(x^4-2x^2+3)^2 (x^4+2x^2+9)>(x^6-x^4+5x^2-9)^2.$$
Moreover, $x^4-2x^2+3=(x^2-1)^2+2>0$ and, if $0<x<1$ then $x^6-x^4+5x^2-9<0$. Thus, if $0<x<1$ then
$(x^4-2x^2+3) \sqrt{x^4+2x^2+9}>-x^6+x^4-5x^2+9$, or equivalently $\alpha(x)>0$. This equality also holds for $x>1$ because both $x^4-2x^2+3$, $x^4+2x^2+9$ and $x^6-x^4+5x^2-9$ are increasing for $x>1$. Therefore if $x>1$ then
$\alpha(x)\ge \alpha(1)=2\sqrt{12}-4>0$. This proves the claim.

Let $x>0$. As
    \begin{eqnarray*}
     h'(x) &=& \phi(x)\left(4+x^2(x^2+1+ \sqrt{x^4+2x^2+9})-\left(3x^2+1+\sqrt{x^4+2x^2+9}+\frac{x(4x^3+4x)}{2\sqrt{x^4+2x^2+9}}\right)\right) \\
                 &=& \frac{\phi(x)}{\sqrt{x^4+2x^2+9}}\left( (x^4-2x^2+3) \sqrt{x^4+2x^2+9} + (x^2-1)(x^4+2x^2+9)-2x^2(x^2+1)\right) \\
                 &=& \frac{\phi(x)}{\sqrt{x^4+2x^2+9}}\left( (x^4-2x^2+3) \sqrt{x^4+2x^2+9} + x^6-x^4+5x^2-9\right) \\
                 &=& \frac{\phi(x)\alpha(x)}{\sqrt{x^4+2x^2+9}}
    \end{eqnarray*}
we conclude that $h$ is an increasing function, therefore, $h(x)>h(0)=0$. Equivalently
    \begin{equation}\label{psiMayor}
     \psi(x)>\frac{x\phi(x)(x^2+1+ \sqrt{x^4+2x^2+9}}{2}.
    \end{equation}
As the greatest root of the quadratic polynomial $q(T)=T^2-(x^3+x)\phi(x)T-2x^2\phi(x)^2$ is $\frac{x\phi(x)(x^2+1+\sqrt{x^4+2x^2+9})}{2}$, inequality (\ref{psiMayor}) implies that $q(\psi(x))>0$.
Hence
    $$f'(x)=-\frac{2}{x^3}-2\frac{-x^2\phi(x)\psi(x)-\phi(x)\psi(x)-2x\phi(x)^2}{x^2\psi(x)^2} = -\frac{2q(\psi(x))}{x^3\psi(x)^2}<0.$$
Thus $f$ is decreasing for $x>0$, as desired.
\end{proof}

By Lemma~\ref{1/3} we have $0\le \sigma_D^2 \le \frac{t^2}{3}$.

As $\A_t$ and $\D_t$ are independent, the variance of $\AL_t=\A_t+\D_t$ is $\sigma_l^2=\sigma_u^2+\sigma_D^2$, hence,
    $$\sigma_u^2 \le \sigma_l^2\le \sigma_u^2 +\frac{t^2}{3}.$$

Notice that the case $\sigma_l^2= \sigma_u^2 +\frac{t^2}{3}$ occurs for the highest degree of noise, i.e., when $\D_t$ has maximum variance $\sigma_D^2=\frac{t^2}{3}$.

\section{The effects of a genotype s: $\varphi_s$, $\Psi_s^{\epsilon}$}\label{understanding_case-control}

We defined the discrete random variable $\MC{F}$, which takes only two values $0$ and $1$ depending on whether the individual is a control or a case, respectively. Let $M_{0s}$ and $M_{1s}$ be the number of controls and cases with the given genotype $s$, respectively, and let $M_s=M_{0s}+M_{1s}$.
Let $b_s$ denote the theoretical probability of being a case for an individual with genotype $s$ in a GWAS with $N_0$ controls and $N_1$ cases.

We assumed that the presence of the genotype multiply by a factor $\varphi_s$ as the probability of presenting a phenotype. We denoted the probability of having the phenotype for the individuals with genotype $s$ in the general population as $c_s$. Let $p$ be the expected value of $c_s$, which is usually called the prevalence of the phenotype in the general population. Therefore,
    $$p\varphi_s=c_s= P(\MC{F}=1|s) =\frac{P(\MC{F}=1,s)}{P(s)} = \frac{P(\MC{F}=1)P(s|\MC{F}=1)}{P(s)} = \frac{p P(s|\MC{F}=1)}{P(s)}$$
Since controls are a representation of the population and cases are a sample of individuals with the phenotype, then we can compute the following estimations:
    $$b_s \approx\frac{M_{1s}}{M_s}, \quad P(s)\approx \frac{M_{0s}}{N_0} \quad \text{ and } \quad P(s|\MC{F}=1) \approx \frac{M_{1s}}{N_1}$$
Therefore
    $$\varphi_s\approx \frac{N_0}{N_1} \frac{M_{1,s}}{M_{0,s}} = \frac{N_0}{N_1} \frac{M_{1,s}/M_s}{1-M_{1,s}/M_s} =
		\frac{N_0}{N_1} \frac{\widehat{b}_s}{1-\widehat{b}_s}$$
and so
	    $$\widehat{b}_s=\frac{c_sN_1}{p N_0+c_sN_1}=\frac{1}{\frac{p N_0}{c_sN_1}+1}=\frac{\varphi_s N_1}{N_0+\varphi_s N_1}.$$

Next, we proceeded as in the tests statistics defined in (\ref{testbetacool}) for the new $b_s$ that takes into account the effect of the genotype $s$. Therefore, we may construct similar statistics in the same way as above, where the decision rule at a $100(1-\epsilon/2)\%$ are
\begin{equation}\label{testbeta'}
\begin{array}{lcl}
\mbox{ Accept }H_0'&\mbox{ if }& \omega'_{\epsilon/2,\varphi}\leq \widehat{b}_s\leq \omega'_{1-\epsilon/2,\varphi}\\
&&\\
\mbox{ Reject }H_0'& &\mbox{ otherwise }
\end{array}
\end{equation}
 where $Pr(\Beta(\alpha_{M_s}',\beta_{M_s}')<\omega'_{\epsilon/2,\varphi})=Pr(\Beta(\alpha_{M_s}',\beta_{M_s}')>\omega'_{1-\epsilon/2,\varphi})=\epsilon/2$, and
 $$\alpha_{M_s}' = \frac{\varphi N_1}{N_0+\varphi N_1}\left( \frac{M_s(N_0+\varphi N_1-1)}{N_0+\varphi N_1-M_s}-1\right) \quad \text{and} \quad \beta_{M_s}' = \frac{N_0}{N_0+\varphi N_1}\left( \frac{M_s(N_0+\varphi N_1-1)}{N_0+\varphi N_1-M_s}-1\right)$$.

Here, we defined a new parameter, called \emph{the critical effect} of the genotype $s$, with a certain confidence level $\epsilon$ as the effect value, $\varphi_s$, of the expected distribution under the decision rule (\ref{new-testbeta}) where $b_s=q'_{1-\epsilon/2}$ if $b_s>\hat{b}=\frac{N_1}{N}$ or $b_s=q'_{\epsilon/2}$ if $b_s\leq\hat{b}=\frac{N_1}{N}$.

That is, the critical effect size $\Psi_s^\epsilon $ of genotype $s$ at a $\epsilon$ confidence level is the greatest real number $\varphi_s$ (respectively lowest), such that the test given by (\ref{testbetacool}) for $b_s =\frac{\varphi_s N_1}{N_0+\varphi N_1}$ rejects the null hypothesis.

\section{Relationships among the models}\label{Moritz}
Figure~\ref{plots} shows how the models of CT and BT could be connected. Each considered variable ($AA$, $BB$ or $AB$) was treated as a different model. This implies that there is no one-to-one correspondence between CT and BT models. For instance, the allelic CT model is not equal to the probability of being case with the allele $A$ or $B$, as analyzed in BT model. On the contrary, both cases ($A$ and $B$) must be taken into account when compared to the allelic CT model.

Thus, there are eight models in the BT which seems to be of interest: allele A ($A$), allele B ($B$), dominant of $A$ ($AA\cup AB$), recessive of $A$ ($AA$), dominant of $B$ ($AB \cup BB$), recessive of $B$ ($BB$), homozygous ($AA\cup BB$) and heterozygous ($AB$). Thus, summarizing that there are five interesting models for CT against eight meaningful models for the new BT.

\begin{figure}[h!]
    \begin{center}
    \begin{tabular}{c}
	\includegraphics[height=7cm]{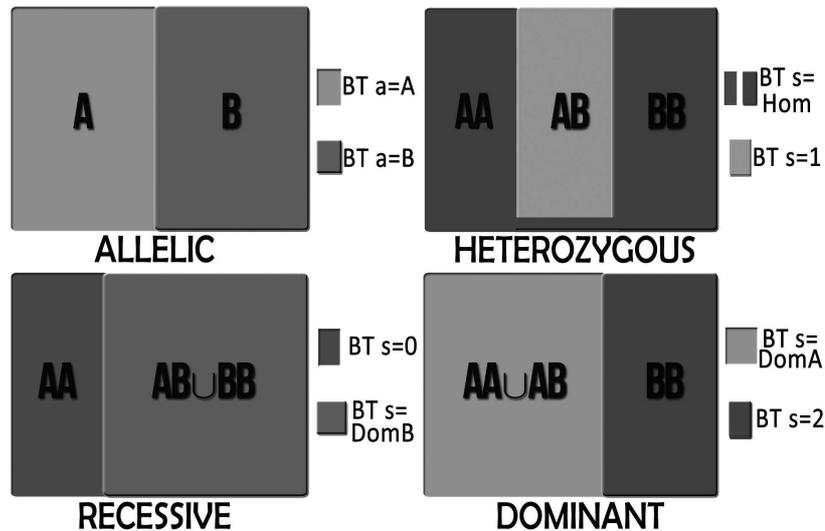}
    \end{tabular}
    \end{center}
\caption{\label{plots} Paired models, \emph{i.e.} relationships between the models of Conventional Test (CT), composed by allelic, heterozygous, recessive, and dominant; the new test proposed, called the Beta Test (BT), composed by allele $A$ (BT $a=A$), allele $B$ (BT $a=B$), $s=AA$ (BT $s=0$), $s=AB$ (BT $s=1$), $s=BB$ (BT $s=2$), dominant of $A$ (BT $s=$DomA), dominant of $B$ (BT $s=$DomB), and homozygous (BT $s=$Hom).}
\end{figure}

\bibliographystyle{plainnat}
\bibliography{garcia-cremades-bib}

\end{document}